\documentclass[%
 reprint,
superscriptaddress,
preprintnumbers,
 aps,
 prd,
 ,
floatfix,
]{revtex4-2}

\usepackage{amsmath, amsthm, amscd, amssymb}
\usepackage{mathtools}
\usepackage{graphicx}
\usepackage{dcolumn}
\usepackage{caption}
\usepackage{subcaption}
\usepackage{bm}
\usepackage[usenames,dvipsnames]{xcolor}
\usepackage{babel}[english]
\usepackage[ruled,linesnumbered]{algorithm2e}
\SetNlSty{mdseries}{}{:}
\usepackage{enumerate}
\usepackage{enumitem}
\newtheorem{theorem}{Theorem}
\newtheorem*{remark}{Theorem}
\usepackage{float}
\usepackage{mathtools}
\usepackage[colorlinks]{hyperref}
\usepackage[nameinlink]{cleveref}

\usepackage[normalem]{ulem} 

\newcommand{\mycomment}[1]{}
\newcommand{\estimateE}[2]{\mathbb{E}_{#1}\left[#2\right]} 
\newcommand{\fkl}{\textrm{KL}(p\,||\,q_\theta)}
\newcommand{\rkl}{\textrm{KL}(q_\theta\,||\,p)}
\newcommand{\normdist}{\mathcal{N}} 
\newcommand{\realset}{\mathbb{R}} 
\newcommand{\id}{\mathbb{I}} 
\newcommand{\vect}[1]{\boldsymbol{#1}} 
\newcommand{\Dphi}{\mathcal{D}[\phi]}
\newcommand{\wbar}{{\bar{w}} }

\DeclareMathOperator*{\argminAlg}{\arg\!\min}

\usepackage{amsthm}
\newtheorem{definition}[theorem]{Definition}

\definecolor{THblue}{rgb}{0.47, 0.62, 0.8}
\newcommand{\tobias}[2][\color{THblue}]{{#1#2}}
\usepackage[normalem]{ulem}

\begin{document}
\long\def\/*#1*/{}

\title{Detecting and Mitigating Mode-Collapse for Flow-based Sampling of Lattice Field Theories}

\author{Kim A. Nicoli}
\email{kim.a.nicoli@gmail.com}
\affiliation{Transdisciplinary Research Area (TRA) Matter, University of Bonn, Germany}
\affiliation{Helmholtz Institute for Radiation and Nuclear Physics (HISKP), Bonn, Germany}
\affiliation{Berlin Institute for the Foundations of Learning and Data (BIFOLD), Berlin, Germany}

\author{Christopher J. Anders}
\affiliation{Berlin Institute for the Foundations of Learning and Data (BIFOLD), Berlin, Germany}
\affiliation{Machine Learning Group, Technische Universit\"{a}t Berlin, Berlin, Germany}

\author{Tobias Hartung}
\affiliation{Northeastern University - London, London, United Kingdom}

\author{Karl Jansen}
\affiliation{CQTA, Deutsches Elektronen-Synchrotron DESY, Zeuthen, Germany}

\author{Pan Kessel}
\email{pan.kessel@gmail.com}
\affiliation{Prescient Design, gRED, Roche, Basel, Switzerland}

\author{Shinichi Nakajima}
\affiliation{Berlin Institute for the Foundations of Learning and Data (BIFOLD), Berlin, Germany}
\affiliation{Machine Learning Group, Technische Universit\"{a}t Berlin, Berlin, Germany}
\affiliation{RIKEN Center for AIP, Tokyo, Japan}


\begin{abstract}
We study the consequences of mode-collapse of normalizing flows in the context of lattice field theory. Normalizing flows allow for independent sampling. For this reason, it is hoped that they can avoid the tunneling problem of local-update MCMC algorithms for multi-modal distributions. In this work, we first point out that the tunneling problem is also present for normalizing flows but is shifted from the sampling to the algorithm's training phase. Specifically, normalizing flows often suffer from mode-collapse for which the training process assigns vanishingly low probability mass to relevant modes of the physical distribution. This may result in a significant bias when the flow is used as a sampler in a Markov-Chain or with Importance Sampling. We propose a metric to quantify the degree of mode-collapse and derive a bound on the resulting bias. Furthermore, we propose various mitigation strategies in particular in the context of estimating thermodynamic observables, such as the free energy. 
\end{abstract}

\maketitle

\section{Introduction}\label{sec:intro}
Using normalizing flows for sampling in lattice field theory has gained significant attention over the last few years. Several works have been carried out in the domain of scalar field theories~\cite{albergo2019flow,nicoli2021estimation,hackett2021flow,nicoli2021machine,gerdes2022learning,caselle2022stochastic,caselle2022stochasticproceedings,singha2022conditional,albandea2022learning,pawlowski2022flow}, $U(1)$~\cite{flowsforlattice2,finkenrath2022tackling} and $SU(N)$~\cite{flowsforlattice3,PhysRevLett.128.032003,PhysRevLett.128.032003,bacchio2022learning} pure gauge theories, and fermionic gauge theories~\cite{albergo2021flow,abbott2022gauge}. This rapid development is attributed to the appealing conceptual properties of flow-based sampling. A well-trained flow approximately acts as a trivializing map \cite{luscher2010trivializing} and therefore can significantly reduce the integrated autocorrelation time of physical observables. The practical obstruction to harnessing this conceptual advantage is that the training process becomes increasingly challenging as the dimensionality of the lattice increases, resulting in poor volume scaling \cite{del2021machine,del2021efficient,de2021scaling,abbott2022aspects}. Furthermore, it is well-known that generative models struggle to learn long-range correlations \cite{bengio1994learning} which is crucial as a critical point is approached. When the continuum limit of the theory is taken, both challenges manifest simultaneously: the volume needs to be increased as the critical point is approached. As a result, it remains an open question whether useful architectures can be found for addressing critical slowing down in the continuum limit.

Another conceptually appealing property of normalizing flows is that they allow for independent sampling, thus making flow densities suitable for being combined with Metropolis-Hastings accept-reject schemes. This approach is often referred to as Neural-MCMC~\cite{nicoli2020asymptotically,albergo2019flow,doi:10.1073/pnas.2109420119,grenioux2023sampling}. As a result, it may be hoped that they can avoid the tunneling problem which arises when local update MCMC algorithms are applied to theories that have degenerate minima separated by high action barriers. However, normalizing flows are typically trained by self-sampling in the context of lattice field theory \cite{albergo2019flow}. As we will discuss, this bears the risk that the training will assign vanishing low probability mass to some of the modes of the theory \cite{nicoli2021machine, nicoli2021estimation, hackett2021flow}, since the training objective will not strongly penalize this. If mode-collapse happens, certain modes of the theory will not be probed by the sampler. This problem, therefore, leads to substantially biased estimators of physical observables as shown in \cref{fig:sketchrevKL}.

In our work we study mode-collapse and the more general mode-mismatch phenomenon, both theoretically and numerically. We first discuss in detail the mode-seeking nature of the standard self-sampling-based training procedure which corresponds to minimizing the reverse Kullback-Leibler (KL) divergence~\cite{hackett2021flow}. We compare this to an alternative training procedure which is based on minimizing the forward (as opposed to the reverse) KL divergence and review why it is equivalent to maximum likelihood training. This objective has the advantage that it is substantially less vulnerable to mode-collapse but has the disadvantage that it requires representative configurations sampled from the theory. In many applications, this prevents this objective from being of any use since if such configurations are available, we can directly measure physical observables on them and a flow is not necessary. However, we point out that there is an important exception to this: for thermodynamic observables, such as the free energy, it is still useful to train a flow. This is because these observables are typically obtained by integration through the parameter space of the theory and thus require a significant number of Markov chains along a discretized trajectory in the parameter space. By training a flow on samples generated at a single point in parameter space, we can completely avoid the need for these additional Markov chains. In this important scenario, it is thus sensible and, as we argue, advisable to use forward KL training for the flow to significantly reduce mode-collapse. Besides modifying the training procedure, we also propose to mitigate mode-collapse by combining two flow-based estimators for the free energy. As a side remark, we note that concurrent works have been proposing strategies, alternative to the Forward KL objective, trying to mitigate mode collapse. These include more stable path gradient estimators~\cite{pmlr-v162-vaitl22a, Vaitl2022GradientsSS}, learning deformed target distributions~\cite{mate2023learning} and annealed importance sampling~\cite{midgley2022flow}.

We then study the bias induced by mode-collapse theoretically. Specifically, we derive a bound on the bias of the estimator for physical observables. This allows us to propose a natural metric to quantify the degree of mode-collapse of the sampler. 

The effectiveness of our proposed methods is then demonstrated on a two-dimensional $\phi^4$ scalar theory.

We stress that our study focuses on the estimation of the free energy, an example of thermodynamic quantity involving the partition function~\footnote{Other examples of such thermodynamic observables not directly accessible with HMC are, for instance, entropy and pressure.}, a crucial subset of physical observables in lattice field theory \cite{caselle2018qcd}.
Estimating these observables with standard Markov-Chain-based methods requires sampling configurations at many different values in parameter space and integrating free energy differences from a known reference value, see previous works for more details~\cite{nicoli2021machine,nicoli2021estimation}. This approach is computationally expensive since it often requires a significant number of HMC chains along the trajectory in parameter space, and crucially leads to high uncertainty, as errors from each chain accumulate upon integration. This problem becomes more severe when one needs to cross a phase transition. There, integrated autocorrelation times explode, thus resulting in larger errors for each Markov chain.  
For this reason, training a normalizing flow using a forward KL objective can often be advantageous: training a normalizing flow requires samples from only a single Markov chain at the target point in parameter space and thus allows us to circumvent the need for any additional chains along the trajectory through parameter space.

We emphasize that the intricacies of training a normalizing flow for multimodal distributions in the context of lattice field theories have been already discussed in \cite{hackett2021flow}. 
Our work builds on this reference but is different in the sense that we consider thermodynamic observables. As explained above, these observables cannot be estimated on the Markov-Chain samples at the target point without the need for additional Markov chains for different coupling values. As a result, training of normalizing flows using the forward KL objective is particularly natural for the estimation of thermodynamic observables.

\begin{figure}[t]
     \centering
     \includegraphics[width=0.99\columnwidth]{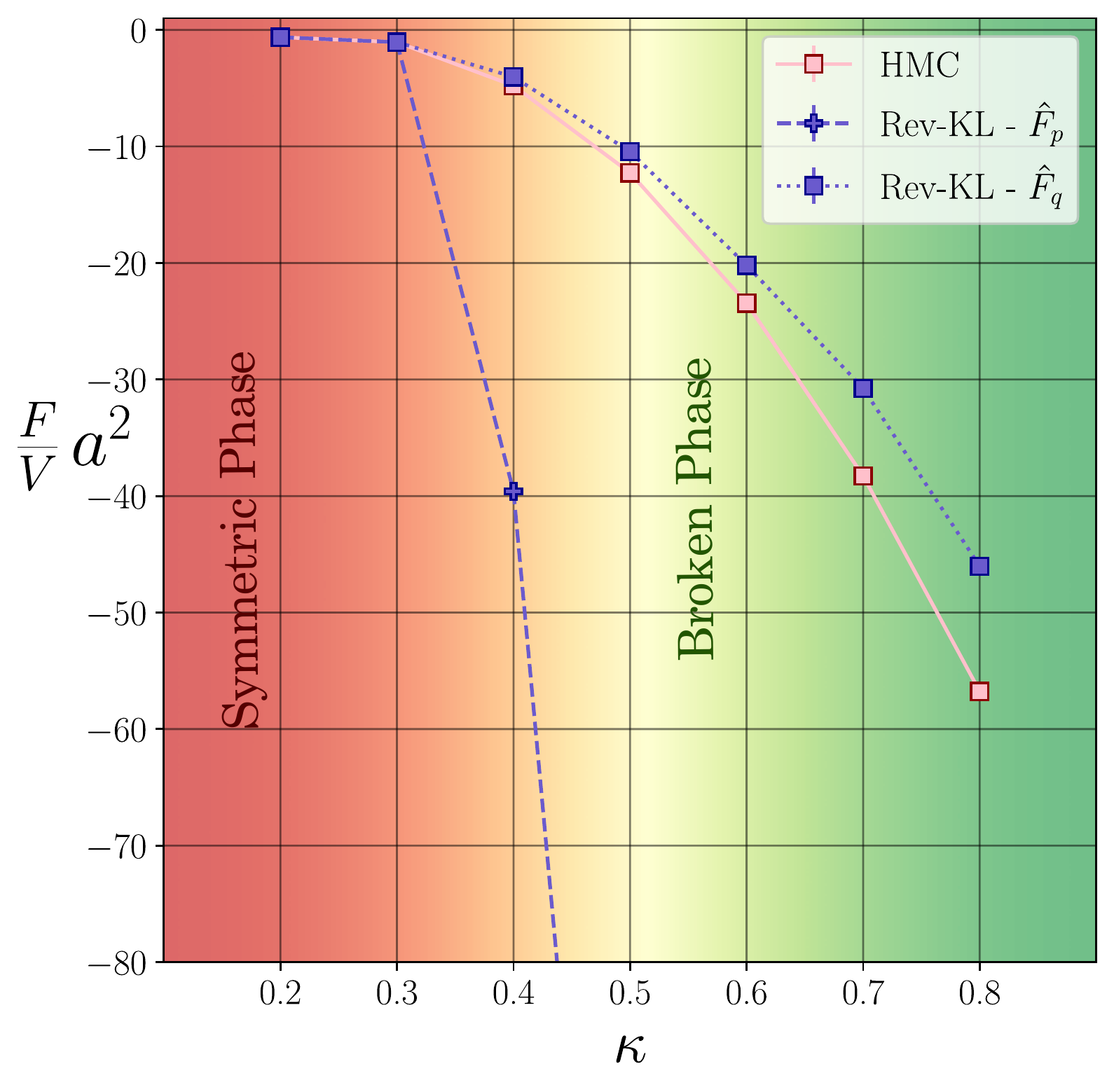}
     \caption{\textbf{Estimation of free energy density in broken and symmetric phases using a reverse-KL trained flow} - estimation of the free energy density using the approach proposed in~\cite{nicoli2021estimation} and \cite{nicoli2021machine}. The second-order phase transition is represented with a color gradient from red to green in the background. The free energy density is estimated using samples drawn from the flow and the target distribution respectively. The flow estimates (purple) are compared to the HMC baseline (pink). The experiments consider lattices $\Lambda=64\times 8$ at fixed coupling $\lambda=0.022$ for the $\phi^4$-theory as in~\cite{nicoli2021estimation}. In the broken phase (green) of the theory, e.g., $\kappa\geq0.3$, the target density has two modes with very little tunneling probability between the two. Using a deep generative sampler, i.e., a flow, trained with reverse-KL in the broken regime leads to biased estimates. This problem lays the foundation of the study presented in this work. }\label{fig:sketchrevKL}
\end{figure}
\section{Training a Generative Model}
\label{sec:normflows}
Normalizing flows~\cite{rezende2015variational,kobyzev2020normalizing,nfreview} are a particular class of generative models giving access to an analytic form of the likelihood.
While this work focuses on flows for concreteness, we stress that the theoretical arguments made in the following sections hold for any generative model allowing for exact likelihood estimation. For such models, a variational density $q_\theta$, the \textit{sampler}, parameterized by a set of weights $\theta$, is optimized to approximate the target density of the lattice field theory
\begin{align}
    p(\phi) = \frac{1}{Z} \exp(-S(\phi)) \,,
\end{align}
where $Z$ is the partition function and $S$ is the action of the theory. 

During the training of a normalizing flow, an efficient transformation to map a base density $q_z$ into a non-trivial target is learned. In practice, the base distribution is chosen such that it allows for efficient sampling. Common choices for the base density are therefore normal or uniform distributions.

The flow uses a diffeomorphism $f$ between the base space $\mathcal{Z}$ and the configuration space $\mathcal{X}$ hence
\begin{align}
    f_\theta:\,z\in \mathcal{Z}\sim q_z \to x=f_\theta(z)\in\mathcal{X} \sim q_\theta\,.
\end{align}
The diffeomorphism $f_\theta$ is a composition of bijective transformations $f_\theta^i$ referred to as \textit{coupling blocks}. Each of these blocks satisfies the following requirements:
\begin{enumerate}
    \item $f_\theta^i$ is a bijection,
    \item both $f_\theta^i$ and its inverse are in $\mathcal{C}^\infty$.  
    \item the determinant of the Jacobian is efficient to evaluate.
\end{enumerate}
The inverse of the transformation $f^{-1}_\theta(x)=z$ therefore always exists by construction. Leveraging these properties, an analytic expression for the likelihood of the flow-based model reads
\begin{align}
    q_{{\theta}}(x) &= q_z(f^{-1}_{{\theta}}(x)) \, \left| \frac{\textrm{d}f_{{\theta}}}{\textrm{d}{z}} \right|^{-1} = q_z({z}) \, \left| \frac{\textrm{d}f_{{\theta}}}{\textrm{d}{z}} \right|^{-1}\,.
\end{align}
Different coupling blocks satisfying the requirements above have been proposed; these include Non-Linear Independent Component Estimation (NICE)~\cite{dinh2014nice}, Real Non-Volume Preserving (RealNVP)~\cite{dinh2016density}, and Generative flow (GLOW)~\cite{kingma2018glow}. We refer to~\cite{nfreview,kobyzev2020normalizing} for an overview of the existing coupling blocks and further technical details.

\subsection{The Forward- and Reverse-KL Divergences}
\label{sec:fwdVSrev}
During training, the normalizing flow is optimized by density matching. It is common practice to minimize the so-called KL divergences to this end although other types of generalized divergences can be used \cite{042471a243c145c38cf15ad98717f998,minka2005divergence,amari1985differential,10.1214/aoms/1177729694,JMLR:v12:reid11a,6832827,61115,JMLR:v6:banerjee05b}. As we will discuss in this section, and in \cref{sec:modecollapse}, choosing an appropriate divergence is crucial to ensure successful training. 

\begin{figure}[t]
     \centering
     \includegraphics[width=0.99\columnwidth]{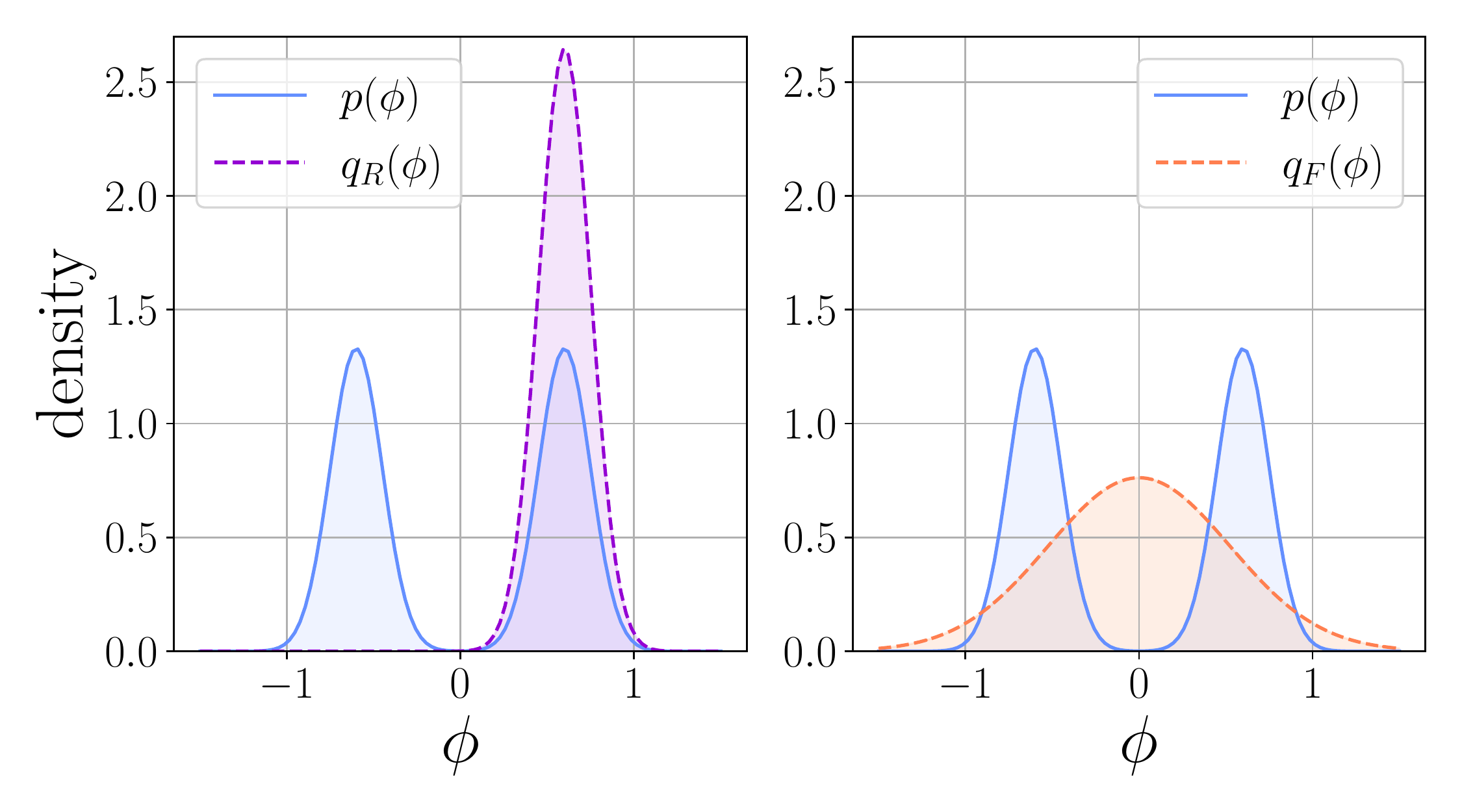}
     \caption{\textbf{Comparison between a reverse-KL (purple) and forward-KL (orange) optimization approach} - a sketch of normalizing flows trained with both KL objectives as described in \cref{sec:fwdVSrev}. Training with a reverse KL shows a mode-seeking behavior and is thus prone to mode-dropping. The forward-KL has instead a mode-covering behavior which leads to larger support over the sampling space.}\label{fig:fwdrevKL}
\end{figure}

The so-called reverse-KL divergence reads
\begin{align}
\label{eq:revkl}
    \rkl &= \int  \mathcal{D}[\phi]\, q_\theta(\phi) \ln \frac{q_\theta(\phi)}{p(\phi)}
\end{align}
where $\mathcal{D}[\phi]$ represents the measure of a high-dimensional integral.
It is worth stressing that the KL divergence is not symmetric hence
\begin{align}
\label{eq:asymmetricKL}
    \rkl \neq \fkl\,.
\end{align}
The right-hand side of \cref{eq:asymmetricKL} is usually referred to as the \textit{forward-KL} which can be written as an expectation value with respect to the target density $p$
\begin{align}
\label{eq:fwdkl}
    \fkl &= \int  \mathcal{D}[\phi] \,p(\phi) \ln \frac{p(\phi)}{q_\theta(\phi)}\,.
\end{align}
These two choices for the divergence lead to different training procedures, as we will discuss in sections \cref{sec:revkl} and \cref{sec:fwdkl}. We also note that the use of reverse and forward KL is not mutually exclusive. In the context of Quantum Chemistry~\cite{noe2019boltzmann}, for instance, a combination of the two is typically chosen.

\subsubsection{Reverse-KL: training by self-sampling}
\label{sec:revkl}
The reverse KL divergence is the standard choice for training normalizing flows on the lattice.
This is because lattice field theory comes with an action $S$ which is known in closed form -- in contrast to many other machine learning applications.  

The reverse KL divergence \eqref{eq:revkl} can be approximated by a Monte-Carlo estimate~\cite{nicoli2020asymptotically,nicoli2021estimation} as follows
\begin{align}
    \rkl &= \estimateE{q_\theta}{\ln \frac{q_\theta(\phi)}{p(\phi)}} \nonumber \\
    & \approx \frac{1}{N} \sum_{i=1}^N \left( S(\phi_i) +\ln q_\theta(\phi_i) \right) + \textrm{const.} \,.
\end{align}
Here, the field configurations are sampled from the flow, i.e. $\phi_i \sim q_\theta$, and thus the training relies on self-sampling. In particular, the partition function $Z$ contributes by a shift term that is constant with respect to the parameters of the flow and thus can be ignored for optimization by gradient descent.

Training using a reverse-KL is therefore very efficient because it does not require samples from the target density $p$ due to self-sampling. Unfortunately, this comes at a cost as this objective is known to be prone to mode-collapse. The sketch on the left-hand side of \cref{fig:fwdrevKL} shows that a reverse-KL-trained flow tends to focus its support on a subset of the modes when the target density is multimodal. This undesirable behavior strongly affects the reverse KL as the dropped modes are not probed in the self-sampling process.

This mode-seeking nature of the reverse KL represents a major drawback and limits the applicability of this framework when the physical density to be learned has more than one mode~\cite{hackett2021flow,nicoli2021machine,mate2023learning}.

\subsubsection{Forward-KL: training by maximum likelihood}
\label{sec:fwdkl}
The forward KL divergences can be written as an expectation value with respect to the target density $p$ and thus also be approximated by Monte-Carlo
\begin{align}
    \fkl &=  \estimateE{p}{\ln \frac{p(\phi)}{q_\theta(\phi)}} \nonumber \\
    &= -\frac{1}{N} \sum_{i=1}^N \log q_\theta(\phi_i) + \textrm{const.}  \,.
\end{align}
In contrast to the reverse KL divergence, the samples are here to be drawn from the density $p$ of the theory, i.e. $\phi_i \sim p$. As can be seen from the equation, the minimization of the forward KL corresponds to maximizing the likelihood of the model. In the machine learning literature, the forward training procedure is thus also known as \emph{maximum likelihood training}. Indeed, this has already been explored in the context of lattice field theory~\cite{hackett2021flow}.

This training procedure has the advantage that it is mode-covering since all modes of the physical target density $p$ will necessarily be probed in training.
It has however the disadvantage that it requires samples from the target $p$. In lattice applications, these are typically generated by a Monte-Carlo algorithm, such as HMC. However, if these configurations are available, one can directly measure physical observables on them and there is therefore no need to train a flow in the first place. 

One may thus wonder if this training procedure is of any use then. For thermodynamic observables however, such as the free energy, one typically does not only require a single Markov chain for the target density $p$ but a whole series of Markov chains along a discretized trajectory in the parameter space of the theory. As we will review in the next section, a flow allows us to completely avoid the need for these additional Markov chains. For the important class of thermodynamic observables, forward KL training is thus well-justified and, as we will show, advisable. 
 
\section{Reliable Estimators in Presence of Mode-Collapse}
\label{sec:modecollapse}
Combining deep generative models, e.g. normalizing flows, with neural importance sampling (NIS) has been shown to be a fruitful approach for estimating thermodynamic observables in lattice field theory~\cite{nicoli2021estimation,nicoli2021machine}, statistical mechanics~\cite{nicoli2020asymptotically}, and chemistry~\cite{noe2019boltzmann,wirnsberger2020targeted,kohler2023rigid}. This approach enables direct estimation of the free energy as well as other thermodynamic observables because flow-based sampling allows for estimating those observables at arbitrary points in the parameter space. Remarkably, this is in stark contrast to standard Markov-Chain Monte-Carlo methods which instead require non-trivial integration in the parameter space~\cite{nicoli2021estimation}. More specifically, NIS allows the computation of a direct Monte-Carlo estimate of the partition function which is crucial for many thermodynamic observables~\cite{nicoli2021estimation} such as entropy and free energy 
\begin{align}
\label{eq:freeEn}
    F=-T\ln Z\,.
\end{align}
In the following, we revise two different estimators of the free energy, namely the \textit{p-estimator} and the \textit{q-estimator}. These allow estimation using samples drawn from the target and the generative model respectively~\cite{nicoli2021machine}.
We stress that these estimators are applicable for any generative model that has a tractable likelihood, such as normalizing flows~\cite{kobyzev2020normalizing}, autoregressive neural networks~\cite{vanOord2016pixelrnn}, and diffusion models~\cite{song2020score}.

\subsection{Different Estimators for the Partition Function}\label{sec:F_estimators}
Given a generative model, irrespective of whether it has been trained using reverse-KL or forward-KL, the resulting sampler $q_\theta$ is an approximation for the target density $p$. Leveraging this, recent works proposed to estimate the partition function of a physical system~\cite{nicoli2020asymptotically,nicoli2021estimation} directly at a given point in parameter space. This approach samples lattice configurations from the generative model $q_\theta$ and estimates the partition function $Z$ with a so-called \textit{q-estimator} with
\begin{align}\label{eq:rev_est}
        Z_q \equiv \mathbb{E}_{\phi \sim q_\theta} \left[ \frac{e^{-S(\phi)}}{q_\theta(\phi)} \right] \approx \frac{1}{N} \sum_{i=1}^N \frac{e^{-S(\phi_i)}}{q_\theta(\phi_i)} \equiv \hat{Z}_{q}  \,,  && \phi_i \sim q_\theta \,. 
\end{align}
Alternatively, when samples from the target $p$ are available, i.e., using a thermalized Markov chain, one can estimate the inverse partition function with the so-called \textit{p-estimator}
\begin{align}\label{eq:fwd_est}
        Z^{-1}_p \equiv \mathbb{E}_{\phi \sim p} \left[ \frac{q_\theta(\phi)}{e^{-S(\phi)}} \right] \approx \frac{1}{N} \sum_{i=1}^N \frac{q_\theta(\phi_i)}{e^{-S(\phi_i)}} \equiv \hat{Z}_p^{-1} \,, && \phi_i \sim p\,.
\end{align}
Combining these results with \cref{eq:freeEn}, one immediately derives corresponding estimators for the free energy
\begin{align}
    \hat{F}_q = - T \log (\hat{Z}_{q}) \label{eq:estq} && \approx && F_q = -T \log (Z_q) \,,\\
    \hat{F}_p = T \log (\hat{Z}_p^{-1}) && \approx && F_p = T \log(Z_p^{-1}) \,. \label{eq:estp} 
\end{align}
Both the p-estimator and the q-estimator can be shown to be asymptotically consistent under the assumption that the supports of the flow $q_\theta$ and the target density $p$ match, i.e. $\textrm{supp}(q_\theta) = \textrm{supp}(p)$ \cite{nicoli2020asymptotically}.
By construction, the learned density, $q_\theta$, has full support over the entire domain of the base distribution -- at least from a purely theoretical point of view. This implies $\textrm{supp}(q_\theta) = \textrm{supp}(p)$ always holds in theory. However, in practice it is not unlikely to have regions of the domain where the density $q_\theta(\phi)$ is vanishingly small. Hence, for a finite number of samples, it can effectively be zero. This leads to incorrect estimation of expectation values of physical observables for any reasonable number of samples $N$.
Furthermore, ensuring that a normalizing flow is invertible also in practice, i.e. to numerical precision, can be very challenging~\cite{pmlr-v130-behrmann21a}.

To analyze the resulting implications for the estimation process, it is useful to define the following generalized notion of the support of the variational density
\begin{definition}
The \textit{effective support of the variational density $q_\theta$ relative to $p$} is given by
\begin{align}\label{eq:eff_rel_supp}
    \widetilde{\mathrm{supp}}_{p,\epsilon}(q_\theta)=\{\phi\in\mathrm{supp}(q_\theta);\ q_\theta(\phi) > \epsilon p(\phi) \} \,,
\end{align}
for a given numerical threshold $\epsilon$. 
The mode dropping set is then given by 
\begin{align}
\mathcal{S} := \mathrm{supp}(p) \setminus \widetilde{\mathrm{supp}}_{p,\epsilon}(q_\theta) \,.
\end{align}
\end{definition}
This definition is useful for the following reason: if the flow is effectively mode dropping, i.e., the mode-dropping set $\mathcal{S}$ is non-empty, the importance weighted estimator, with a finite number of samples $N$, will miss a contribution from the mass $\int_{\mathcal{S}}p(\phi)d\phi$ with approximately the probability $ 1 - \epsilon N\int_{\mathcal{S}}p(\phi)d\phi$.
We note that defining the threshold for the effective support \textit{relative} to the target distribution is pivotal. This is because the absolute definition $q(\phi)<\epsilon$ would have no meaning with regard to mode-dropping. As an example, suppose we have an area of size $O(1/\epsilon)$ and $p=q\leq\epsilon$ in that area. Then the corresponding area with $O(1)$ probability mass would be considered ``$\epsilon$-mode-dropped'' even though $q$ is an exact copy of $p$. We therefore choose a definition for which mode-dropping only exists when $q$ approximately vanishes \textit{relative} to $p$.

It is also useful to define the \textit{effective sampler distribution}
\begin{equation}\label{eq:eff_sampler_dist}
    \widetilde{q}_\theta (\phi) =
    \begin{dcases}
        q_\theta(\phi) / \zeta & \mbox{if } \phi \in \widetilde{\mathrm{supp}}_{p,\epsilon}(q_\theta),\\
        0 & \mathrm{otherwise}, \\
    \end{dcases}
\end{equation}
where $\zeta = \int_{\widetilde{supp}_{p,\epsilon}} \mathcal{D} [\phi] q_\theta(\phi) \leq 1$ represents the multiplicative renormalization factor necessary to guarantee the normalization of $\widetilde{q}_\theta$, i.e., the probability mass out of the effective support $\widetilde{\mathrm{supp}}_{p,\epsilon}(q_\theta)$ is redistributed to the effective support proportionally to the original density $q_\theta(\phi)$. 
With this definition, we express the practical situation where the importance weighted estimator
for a physical observable $\mathcal{O}$ \emph{typically} misses the contribution from the mode-dropping set $\mathcal{S}$, as the assumption that the following approximation holds:
\begin{align}
\hat{\mathcal{O}}
&\equiv  \frac{1}{N}\sum_{i=1}^N \frac{p(\phi_i)} {{q}_\theta(\phi_i)} \mathcal{O}(\phi_i)
\approx
  \estimateE{\phi\sim\widetilde{q}_\theta}{\frac{p(\phi)}{{q}_\theta(\phi)} \mathcal{O}(\phi)} 
\equiv
   \bar{\mathcal{O}},
    \label{eq:ModeDroppingAssumption}
\end{align}
where $ \phi_i\sim {q}_\theta $, for the sample size large enough for Monte Carlo sampling but not too large to assume that $\zeta^N \approx 1$, i.e., the probability that all $N$ samples drawn from $q_{\theta}$ lie within the effective support is close to one.
Since $q_\theta$ has the full support,  it holds that ${\mathrm{supp}}(\widetilde{q}_\theta) = \widetilde{\mathrm{supp}}_{p,\epsilon}(q_\theta)$.
Throughout the manuscript, we will indicate by a hat a (finite sample) estimator, by a bar the expectation over the effective distribution $\widetilde{q}$ -- which corresponds to the average over \emph{typical} samples -- and by an asterisk the expectation over the original distribution $q$.
 Note that, under our assumption of mode-dropping~\eqref{eq:ModeDroppingAssumption}, i.e., ${\mathrm{supp}}(\widetilde{q}_\theta) \nsupseteq {\mathrm{supp}}(q_{\theta}) = {\mathrm{supp}}(p)$, 
 the typical values of the estimator $\hat{\mathcal{O}} \approx \bar{\mathcal{O}}$ can be significantly different from the true expectation value
 \begin{align}
    \mathcal{O}^* &=
    \estimateE{\phi\sim p}{ \mathcal{O}(\phi)} 
    \label{eq:TrueExpectationObservable}\\
    &=
    \estimateE{\phi\sim {q}_\theta}{\frac{p(\phi)}{{q}_\theta(\phi)} \mathcal{O}(\phi)} 
    \notag\\
    &= \lim_{N \to \infty} \frac{1}{N}\sum_{i=1}^N \frac{p(\phi_i)} {{q}_\theta(\phi_i)} \mathcal{O}(\phi_i),
    && \mbox{ where } && \phi_i\sim {q}_\theta.
\notag
 \end{align}
 \par

A more detailed discussion of the effective relative support is provided in \cref{app:effectivesupp}.

We also remark that $\widetilde{\textrm{supp}}_{q_\theta,\epsilon}(p)$ is also interesting to consider, as $\textrm{supp}(q_\theta)\setminus\widetilde{\textrm{supp}}_{q_\theta,\epsilon}(p)\ne\emptyset$ implies effective \textit{``fake'' modes} that are present in $q_\theta$ but not in $p$.

The following theorem holds:
\begin{theorem}\label{th:free_en}
Suppose the trained model is mode dropping, i.e., 
the approximation~\eqref{eq:ModeDroppingAssumption} holds.
Then the $q$-estimator $\hat{F}_q$ and the $p$-estimator $\hat{F}_p$ for the free energy approximate $\bar{F}_{q}$ and $\bar{F}_{p}$, respectively, being bounds on the true free energy as
\begin{align*}
\bar{F}_{q} \geq F \geq \bar{F}_p \,.
\end{align*}
Furthermore, if $\widetilde{\mathrm{supp}}_{p,\epsilon}(q_\theta) \supseteq \mathrm{supp}(p)$ it follows
\begin{align*}
\bar{F}_{q} = F\,,
\end{align*}
and similarly if $\widetilde{\mathrm{supp}}_{q_\theta,\epsilon}(p) \supseteq \mathrm{supp}(q_\theta)$
\begin{align*}
\bar{F}_{p} = F.\,
\end{align*}
\end{theorem}
We prove in \cref{app:theorem1_proof} that the estimators serve as upper and lower bounds of the free energy.\par
In the presence of mode-collapse, the flow has smaller effective support than the target, i.e., 
\begin{align*}
    \textrm{supp}(p) \nsubseteq {\textrm{supp}}(\widetilde{q}_\theta) \,.
\end{align*}
Crucially, this may also happen when the variational density $q_\theta$ is a very \textit{bad} approximation of the true density $p$, see bottom left of \cref{fig:fwVsrev}. While this is strictly not a common manifestation of mode collapse, the following discussion holds for such badly-trained models where the overlap between the support of $q_\theta$ and $p$ is very small.
In this case, the q-estimator in \cref{eq:estq} may thus lead to (possibly strongly) biased results 
since it may not have full effective support under the assumption that the approximation \eqref{eq:ModeDroppingAssumption} holds.
On the other hand, the  q-estimator has the advantage that it is typically more efficient to sample directly from the flow while the  p-estimator requires the (possibly costly) generation of configurations by a Markov Chain. 
Nevertheless, it is advisable to estimate the free energy with both estimators if there is a risk of mode mismatch and ensure that both lead to consistent results.
The phenomenon of mode-collapse is a widely known issue in the field of density estimation~\cite{goodfellow2016deepchap3, agrawal2020advances, dhaka2021challenges,Dhaka2021ChallengesAO}. In particular, when deploying generative models for physical systems, this becomes crucial as neglecting subsets of the modes of a target density would inevitably lead to highly biased estimation of physical quantities. Moreover, this may sometimes not even be detected unless appropriate estimators are used~\cite{nicoli2021machine}. We want to stress that this problem is not restricted to lattice field theories \citep{hackett2021flow, nicoli2021machine} but is also found within other contexts, such as molecular systems \citep{noe2019boltzmann,wu2020stochastic,midgley2022flow}. Having an estimator which quantifies the amount of probability mass being missed by a variational ansatz is therefore highly desirable for more reliable and unbiased estimation of physical quantities.

When the trained model neglects some modes of the target density, hence missing full effective support over the target domain, estimates of physical observables may be biased. A desirable property of our framework is to detect such bias by providing reliable bounds on the error when the model is mode-dropping. When $q_\theta$ has full effective support on the domain of $p$ the expected value of the importance weights $w(\phi)=\frac{p(\phi)}{q_\theta(\phi)}$ reads
\begin{align}\label{eq:w}
    w^* = \estimateE{q_\theta}{\frac{p(\phi)}{q_\theta(\phi)}} =& \int_{\textrm{supp}(q_\theta)} q_\theta(\phi) \frac{p(\phi)}{q_\theta(\phi)}\, \mathcal{D}[\phi]\nonumber \\
    =& \int_{\textrm{supp}(q_\theta)} p(\phi)\, \mathcal{D}[\phi]=1\,.
\end{align}
This expectation value thus measures the degree to which the support of the target density $p$ is covered by the sampler $q_\theta$. Statistically, in the limit of infinite measurements, $w^*$ is always equal to one. However, if the sampler is mode-dropping, hence the approximation  \eqref{eq:ModeDroppingAssumption} holds, then the estimator $\bar{w}$ will be in $[0, 1]$ providing us with a natural quantity to measure the sampler's ability to probe the entire support of the target density $p$. 

We will now derive an estimator for this expectation value. To this end, we rewrite the above expression as 
\begin{align} \label{eq:moddropestexp}
    \bar{w} \equiv \frac{1}{Z}\mathbb{E}_{\phi \sim \widetilde{q}_\theta} \left[ \frac{e^{ -S(\phi)}}{{q}_\theta(\phi)} \right] \,
\end{align}
and we note that the expectation value is now taken with respect to $\widetilde{q}_\theta$.
As shown in the last section, the partition function $Z$ can be approximated by the p-estimator \eqref{eq:fwd_est} when samples from the target density are available. 
Thus, under the assumption that the approximation \eqref{eq:ModeDroppingAssumption} holds,
the following Monte-Carlo estimate approximates \cref{eq:moddropestexp}, i.e.,
\begin{align}\label{eq:moddropest}
    \bar{w} \approx& \frac{1}{\hat{Z}_{p}} \left( \frac{1}{N}\sum_{i=1}^N \frac{e^{ -S(\phi_i)}}{{q}_\theta(\phi_i)}\right)\nonumber \\ =& \left(\frac{1}{N} \sum_{j=1}^N \frac{{q}_\theta(\phi_j)}{e^{ -S(\phi_j)}}\right) \left( \frac{1}{N}\sum_{i=1}^N \frac{e^{ -S(\phi_i)}}{{q}_\theta(\phi_i)}\right) \equiv \hat{w}\,,
\end{align}
where $\phi_i \sim {q}_\theta$ and $\phi_j \sim p$ are sampled from the flow and the target density $p$ respectively.

\subsection{Bounding the Bias of Physical Observables}
\label{sec:biasbound}
Following~\cite{nicoli2021estimation}, given a physical observable $\mathcal{O}$, our goal is to compute the importance weighted estimator $\hat{\mathcal{O}}$, defined in the left-hand side of \eqref{eq:ModeDroppingAssumption}, which approximates the expectation value $\bar{\mathcal{O}}$ over the effective sampler distribution \eqref{eq:eff_sampler_dist}.  
This estimator is not necessarily unbiased to the true value \eqref{eq:TrueExpectationObservable}, if the model ${q}_\theta$ is affected by mode-collapse, i.e., $\textrm{supp}(p) \nsubseteq {\textrm{supp}}(\widetilde{q}_\theta)$ for the approximation~\eqref{eq:ModeDroppingAssumption} to hold. 
Similarly, the bias of the estimator evaluated over a finite number of trials should approximate 
\begin{align}
    \vert \bar{\mathcal{O}} - \mathcal{O}^*\vert = \Big\vert \int \left( 1 - 1_{{\textrm{supp}}(\widetilde{q}_\theta)} (\phi) \right) \,\mathcal{O}(\phi)\,p(\phi)\,\Dphi\Big\vert\,,
\end{align} 
i.e. the bias arises due to the insufficient effective support of the sampler.

\begin{figure}[t]
  \centering
  \subfloat[]{\includegraphics[width=0.99\columnwidth]{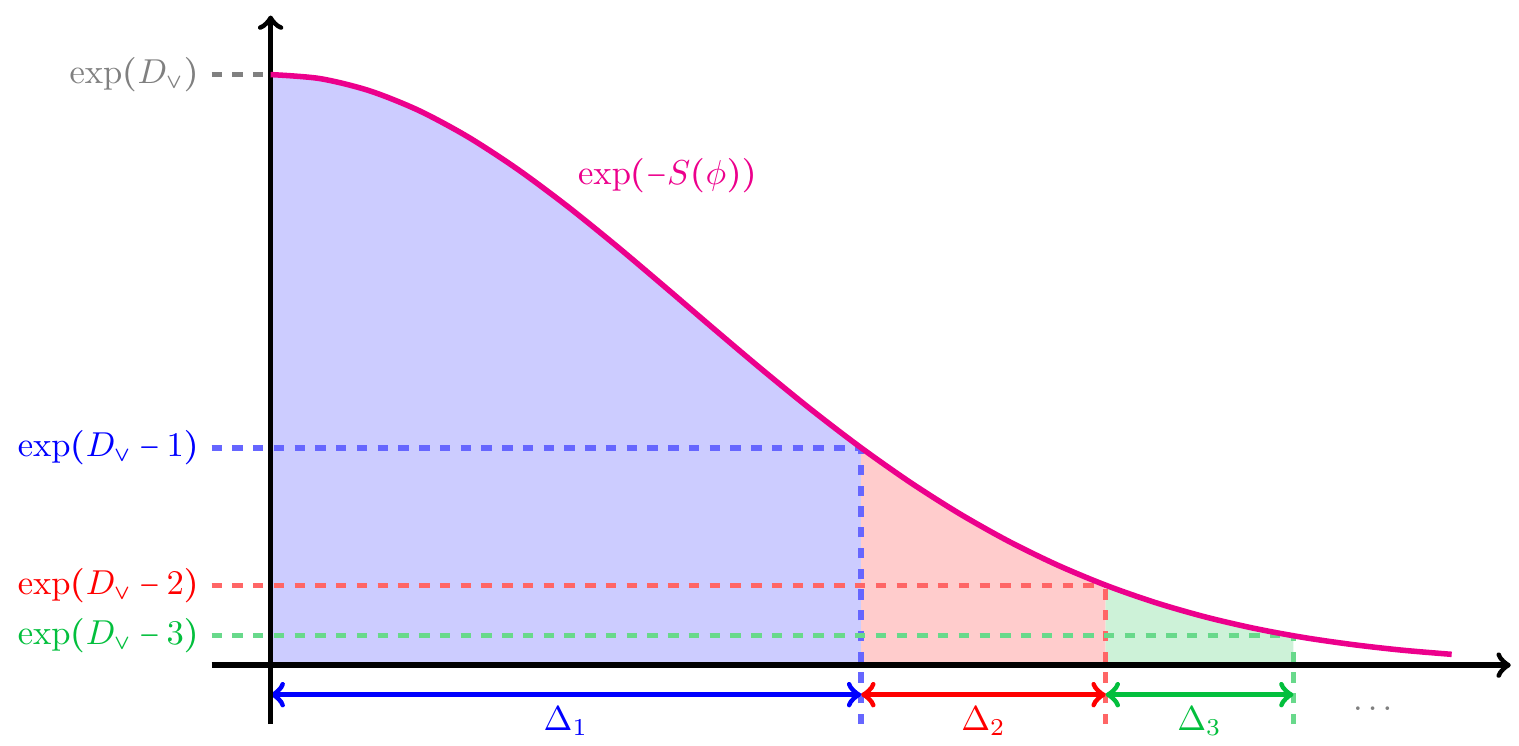} \label{fig:foliation1}} \\
  \subfloat[]{\includegraphics[width=0.99\columnwidth]{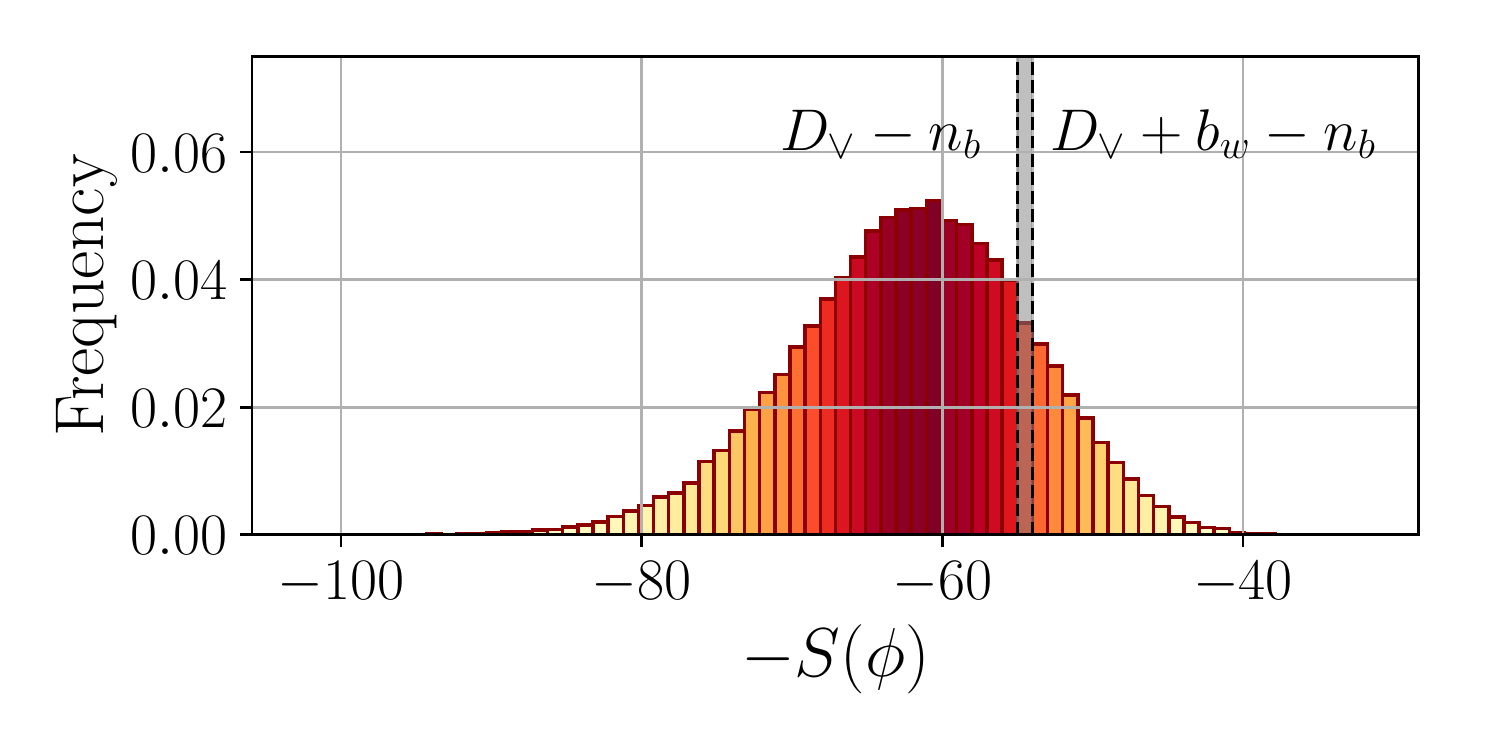} \label{fig:foliation2}}
  \caption{\textbf{Foliation} - (a) visualization of the foliation according to \cref{eq:foliation1} showing $\Delta_1$, $\Delta_2$, and $\Delta_3$ explicitly. (b) Field configurations are distributed into buckets. The labels $n_b$ and $b_w$ thereof represent the bucket index and the width of each bucket, respectively.}\label{fig:foliation}
\end{figure}

\section{The Mode-Dropping Estimator}
\label{sec:mdest}
\begin{figure}
     \centering
     \includegraphics[width=0.99\columnwidth]{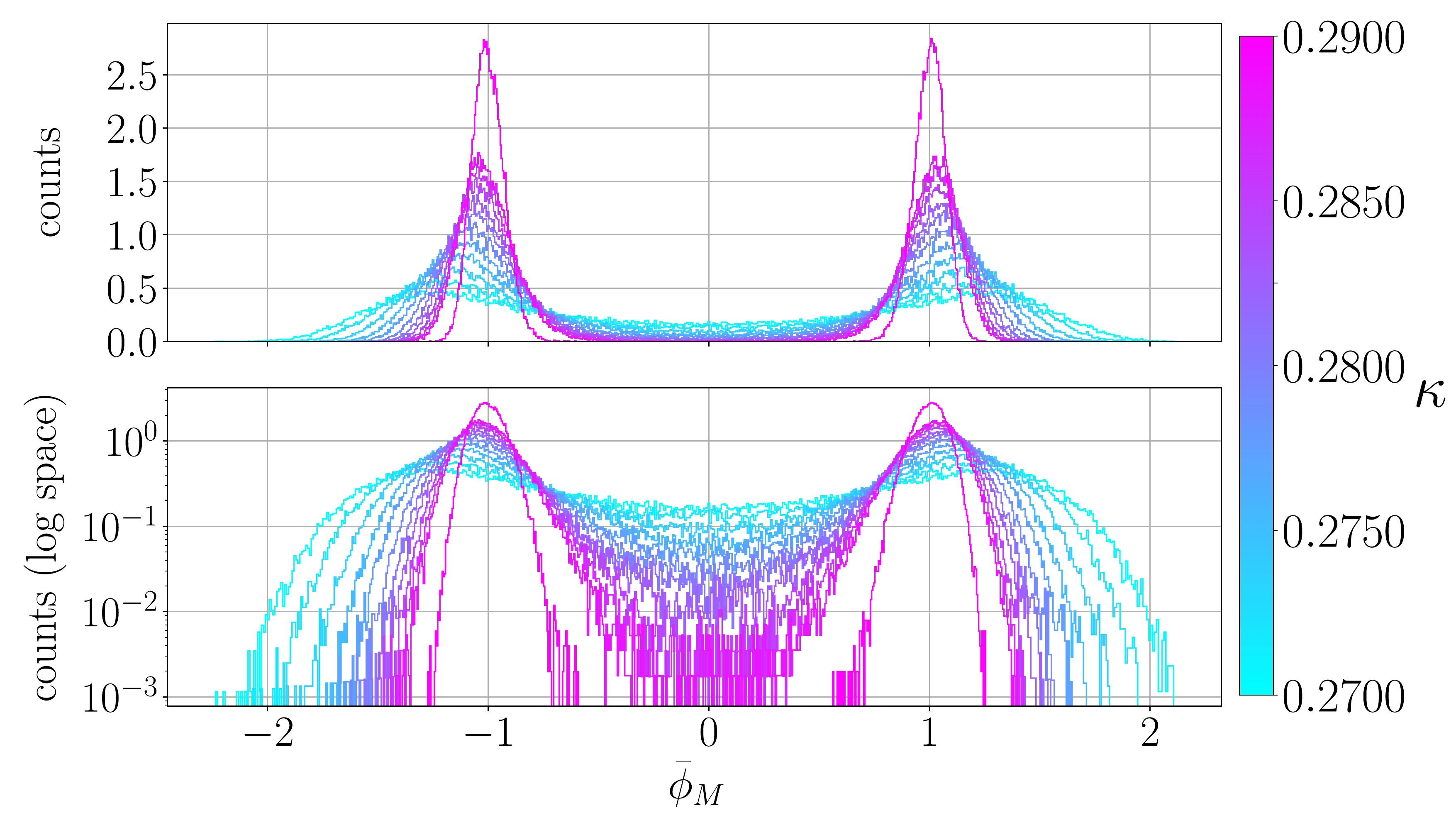}
     \caption{\textbf{Spontaneous symmetry breaking} - spontaneous symmetry breaking for the $\phi^4$ as a function of the hopping parameter using the action in \cref{eq:action}. The top and bottom figures show a histogram of the magnetization -- self-normalized by its absolute value -- in both linear and log-scale respectively. The self-normalization makes the modes centered around +1 and -1. From those histograms we thus observe that the probability between the modes reduces as the hopping parameter $\kappa$ increases.}\label{fig:symmetrybreaking}
\end{figure}
In the following, we aim to derive a bound on this bias. In the process, we will also obtain a natural measure for the degree of mode-collapse. To this end, we foliate the sampling space by disjunct sets
\begin{equation}\label{eq:foliation1}
    \Delta_n := \{\phi \,\, | \,\, e^{D_\vee +1 -n} \geq e^{-S(\phi)} > e^{D_\vee - n}\}
\end{equation}
for $n \in \mathbb{N}$ and $D_\vee \in \mathbb{R}_{\ge0}$. We refer to \cref{fig:foliation} for a visual illustration of the foliation. 
For fixed $n\in \mathbb{N}$, we can define the following weights
\begin{align}
\alpha_n =& \int_{\Delta_n} \left( 1 - 1_{{\textrm{supp}}(\widetilde{q}_\theta)} (\phi) \right) p(\phi) \, \Dphi \nonumber \\
=& \frac{1}{Z}\left(\int_{\Delta_n} e^{-S(\phi)}  \, \Dphi -  \int_{{\Delta_n}\cap{{{\textrm{supp}}(\widetilde{q}_\theta)}}} e^{-S(\phi)}  \, \Dphi \right) \,. \label{eq:alpha_n}
\end{align}
Leveraging these definitions, we derive a bound on the bias in \cref{app:theorem2_proof} which is summarized in the following theorem.
\begin{theorem}\label{th:bias}
Let the action $S$ of the theory and the observable $\mathcal{O}$ be polynomially bounded, i.e.
\begin{align}
    C_\vee \vert\vert \phi \vert\vert^{\alpha_\vee} -D_\vee \leq S(\phi) < C_\wedge \vert\vert \phi \vert\vert^{\alpha_\wedge} + D_\wedge
\end{align}
for some $C_{\vee, \wedge}, D_{\vee, \wedge}, \alpha_{\vee, \wedge} \in \mathbb{R}_{\ge0}$ and
\begin{align}\label{eq:gen_bound_main}
    \vert\mathcal{O}(\phi)\vert \leq c\vert\vert\phi\vert\vert^\alpha + d
\end{align}
for some $c,\alpha,d \in \mathbb{R}_{\ge0}$.
The bias, i.e., the difference between the expectation value $\bar{\mathcal{O}}$ and the true value $\mathcal{O}^*$, then satisfies
\begin{align}\label{eq:bound_maintext}
    \vert \bar{\mathcal{O}} - \mathcal{O}^*\vert \leq \sum_{n\in \mathbb{N}}  \sup_{\phi\in\Delta_n}\left\vert\mathcal{O}(\phi)\right\vert \cdot \alpha_n\,.
\end{align} 
\end{theorem}
The bias is therefore bounded by a weighted sum over the $\alpha_n$. The weighting of each summand depends on the observable $\mathcal{O}$ of interest.
We note that the present discussion is only relevant for non-compact variables. Indeed, continuous functions living on \textit{compact} manifolds $\mathcal{M}$ are integrable, and therefore no additional care is required to handle indefinite forms of the type $0\cdot \infty$. In particular, it follows that for a compact variable $\mathcal{O}$ the bound is straightforward
\begin{align}\label{eq:bound_maintext_compact}
    \vert \bar{\mathcal{O}} - \mathcal{O}^*\vert \leq   \sup_{\phi \in \mathcal{M}}\left\vert\mathcal{O}(\phi)\right\vert \cdot \bar{w}\,.
\end{align}\par 
We note that many physical observables are simple powers of fields, i.e. $\mathcal{O}(\phi) = ||\phi||^k$ for $k \in \mathbb{N}$.
It can be shown that the foliation \eqref{eq:foliation1} along with the polynomial bound of the action $S$ implies that
\begin{align}
    \vert\vert\phi\vert\vert < u_n \,,
\end{align}
where we have defined $u_n = \left(\frac{n}{C_\vee}\right)^{\frac{1}{\alpha_\vee}}$. We refer to \cref{app:boundconfig} for more details. For such observables, the bias can thus be bounded by 
\begin{align}
    \vert \bar{\mathcal{O}} - \mathcal{O}^*\vert \leq \sum_{n\in \mathbb{N}}  u_n^k \,  \alpha_n\,.
\end{align} 
The theorem also naturally relates to the quantity $\bar{w}$ introduced in the last section which quantifies the degree of mode-collapse. In order to provide a single number for the degree of the mode-collapse of the sampler, it is natural to choose a uniform, i.e. observable agnostic, weighing. It then follows from the definition of the $\alpha_n$, see \eqref{eq:alpha_n}, that this weighting measures the mismatch in support between the sampler and the target density
\begin{align}\label{eq:sum_alphan}
    \sum_{n\in \mathbb{N}} \alpha_n = \int \left( 1 - 1_{{\textrm{supp}}(\widetilde{q}_\theta)} (\phi) \right) p(\phi) \, \Dphi =   1 - \bar{w}\,
\end{align}
and is thus directly related to the mode-dropping estimator $\bar{w}$ derived in the last section.
\section{Numerical Experiments}
\label{sec:results}
\begin{figure*}[t]
     \centering
     \includegraphics[width=0.90\textwidth]{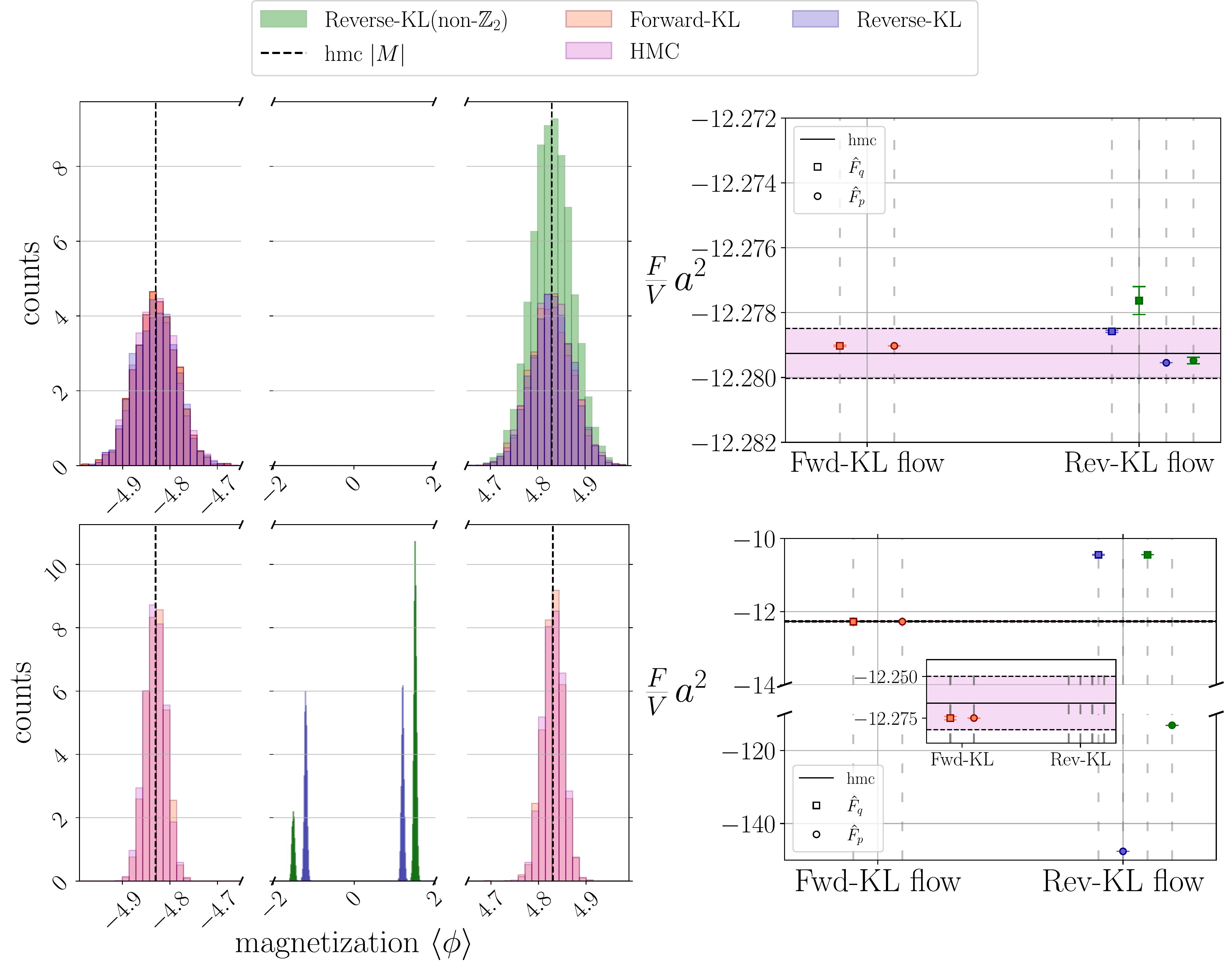}
     \caption{\textbf{Histogram of magnetization (left) and free energy estimates (right) for flow-based models trained with forward- and reverse {KL} at $\kappa=0.5$ for $\Lambda=16\times 8$ (top) and $\Lambda=64\times 8$ (bottom)} - The left-hand side shows histograms of the magnetization for configurations sampled from a forward-KL trained normalizing flow (orange), a $\mathbb{Z}_2$ equivariant reverse-KL flow (purple) a simple non-$\mathbb{Z}_2$-equivariant reverse-KL flow (green) and an overrelaxed HMC (pink). The same choice of colors applies to the right hand side plots. Models were used to sample configurations at $\kappa=0.5$, i.e., well in the broken phase. Dashed black lines are the expected values for the absolute magnetization (and its negative value) using {HMC}. The right-hand side plot shows estimates of free energy densities using the estimators introduced in \cref{sec:modecollapse} for both flow models. Their estimates are computed and then compared to the reference HMC baseline (denoted by the solid black line while the pink surroundings show the standard deviation) and following the method discussed in \cite{nicoli2021estimation} and \cref{app:technical_details}. The q- and the p-estimators for the free energy are shown using different markers. Error bars -- based on the standard deviation -- for the flow-based estimate of the free energy density are often not visible as they tend to be several orders of magnitudes smaller than the plot scale. For the $\Lambda=64\times 8$ case the mismatch between the reverse-KL flows and the {HMC} histograms, shown in the left-hand side plot, is reflected in highly biased estimates (dark blue markers).}\label{fig:fwVsrev}
\end{figure*}
We evaluate our proposed methods to detect and mitigate mode-collapse using the two-dimensional scalar $\phi^4$-theory with action
\begin{align}\label{eq:action}
S[\phi] = \sum_{x \in \Lambda} - 2 \kappa \sum_{\hat{\mu}=1}^2 \phi(x) \phi(x + \hat{\mu}) + (1 - 2 \lambda)& \phi(x)^2 \nonumber \\
+ &\lambda \, \phi(x)^4 \,.
\end{align}
where  $\lambda$ is the bare coupling while $\kappa$ is the hopping parameter. We refer to \cite{nicoli2021estimation} for more details on this hopping parameterization of the action. Throughout all our experiments we keep the bare coupling fixed at $\lambda=0.022$ and vary $\kappa$ such that the theory crosses the phase transition due to the spontaneous breaking of its $\mathbb{Z}_2$ symmetry, i.e. $\phi \to - \phi$. As the hopping parameter $\kappa$ increases, spontaneous magnetization is observed. This is illustrated in \cref{fig:symmetrybreaking} for which the hopping parameter takes values through the critical region around $\kappa_c\approx0.275$. The curves show the density (top) and log density (bottom) of the normalized magnetization with different colors referring to different values of the hopping parameter~$\kappa$. Spontaneous symmetry breaking is observed as the distribution of the magnetization changes from a wider single-mode to a bi-modal density with a suppressed tunneling probability between the two modes. This suppression is accentuated as the value of $\kappa$ increases.\par

\subsection{Free Energy Estimators}

Our first numerical experiment analyzes the performance of two normalizing flows trained with both objectives described in \cref{sec:fwdVSrev}. We refer to those flows as the forward-KL flow and the reverse-KL flow if they were trained with maximum likelihood or self-sampling respectively. We train for a hopping parameter $\kappa=0.5$ such that the theory is in its broken phase, see \cref{fig:fwVsrev}. For maximum likelihood training, we use 50M samples generated by an overrelaxed HMC. Following \cite{nicoli2021estimation}, we choose an architecture for the (reverse-KL) normalizing flows such that those models are manifestly invariant under $\mathbb{Z}_2$ symmetry (blue). In order to highlight the effects of mode-collapse, we also train reverse-KL flows without the $\mathbb{Z}_2$ inductive bias (green) thus expecting these models to be prone to mode-collapse.

The reference estimates for the true free energy were obtained via HMC simulations. Similarly to the approach followed in \cite{nicoli2021estimation} such estimates are obtained by discretizing the hopping parameter space so that free energy differences can be estimated via HMC along the trajectory. Those contributions are added, integrating such trajectory up to the desired point at which the free energy needs to be estimated. Further technical details on how deep generative models were trained and HMC reference values obtained, can be found in \cref{app:technical_details}.

As can be seen on the left-hand side of \cref{fig:fwVsrev}, the forward-KL flow (orange) very closely reproduces the reference distribution by HMC (pink). For the reverse-KL trained flows, we see that for smaller systems (top row in \cref{fig:fwVsrev}), leveraging the $\mathbb{Z}_2$ inductive bias leads to a good approximation (blue) while the non-$\mathbb{Z}_2$ equivariant flow (green) fails to capture both modes. For larger systems, instead, both $\mathbb{Z}_2$ equivariant, and non-equivariant, flows are not able to capture most of the support of the target density $p$, thereby resulting in poor approximations.

 On the right-hand side of \cref{fig:fwVsrev}, we estimate the free energy density of the system using different flows. We use the same color scheme as on the left-hand side and measure the free energy using both the \textit{p-estimator} (circle) \cref{eq:estp} and the \textit{q-estimator} (square) \cref{eq:estq} of the free energy. Our numerical results indeed agree with the theoretical prediction of \cref{th:free_en}.  
 Specifically,  with the model trained with maximum likelihood, both estimators lead to compatible predictions with the HMC estimator. This is consistent with the left-hand side of the plot which suggests that no mode-collapse took place for this model. 

For the reverse-KL flow, however, such an agreement may not be expected as the left-hand side of \cref{fig:fwVsrev} shows a mismatch in the support. 
The right-hand side plots show that the non-$\mathbb{Z}_2$ equivariant flow (green) in the $\Lambda=16\times 8$ (top row) case is dropping the left-hand mode while its $\mathbb{Z}_2$-equivariant counterpart (blue) covers both modes. Nonetheless, as the dimensionality increases the density estimation task becomes increasingly more challenging thus preventing the $\mathbb{Z}_2$-equivariant reverse-KL flows to train effectively for $\Lambda=64\times 8$.
As a result, we find that the q-estimator overestimates the true value F for both lattice sizes (top and bottom rows), while for the $\Lambda=64\times 8$ case the p-estimator substantially underestimates F, i.e. $\hat{F}_{q_\theta} \geq F \geq \hat{F}_p $, as predicted by \cref{th:free_en}. This latter situation suggests that the reverse-KL flows (green and blue) are limited in approximating of the target density resulting in very different effective supports. While strictly speaking this is not usually referred to as mode-collapse, it can be understood through the same lenses.
 
 These experiments thus illustrate that: a) a mode-covering objective such as the forward-KL is more resilient when the target density is multimodal and shows sparse effective support and b) when the effective support does not match, both p- and q-estimators of the free energy from \cref{sec:modecollapse} give upper and lower bound respectively. Moreover, we note that training using a forward KL objective does not worsen the performance compared to using a reverse KL. Practically, if the variational distribution presents some \textit{``fake'' modes}, $\phi$ s.t. 
 $q_\theta(\phi) / p(\phi) \gg 1$, 
 field configurations sampled in these regions will always be exponentially suppressed in the reweighting phase. We emphasize once again that a significant drawback of training using the pure form of forward KL is the necessity for training samples. Although this limitation applies in general, it does not pose a problem for our specific task of estimating thermodynamic observables.\par
\begin{figure}[t]
     \centering
     \includegraphics[width=0.42\textwidth]{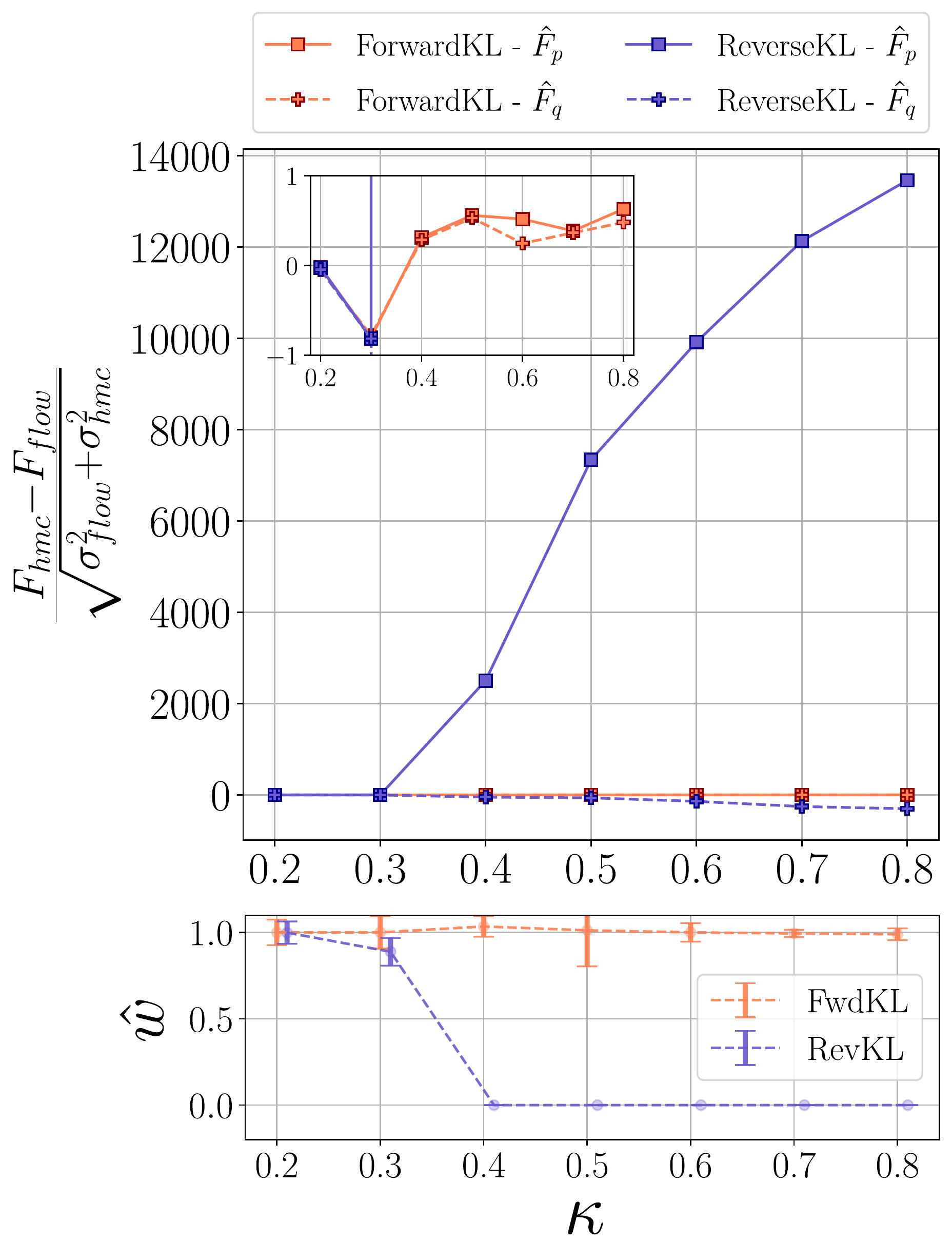}
     \caption{\textbf{Analysis of the free energy estimates for forward and reverse {KL} flows for a $64\times8$ lattice and different hopping parameters $\kappa$} - colors denote flow models trained with different objectives, namely the forward-{KL} (orange) and reverse-{KL} (purple). Markers refer to the two different estimators of the free energy introduced in \cref{eq:estq} and \cref{eq:estp}. Specifically, \textit{plus} and \textit{square} are used respectively. Every point in the inset shows that flow and {HMC} estimates are compatible, i.e., the gap is within the statistical uncertainty. The lower plot relates the results from above to the mode-dropping estimator \cref{eq:moddropest}. This demonstrates that when the flow is a good approximator for $p$, the estimator is close to one. When modes of the distribution are missed, e.g., $q_\theta$ is a bad approximator of $p$, the estimator of $\wbar$ quickly decays to zero.}\label{fig:absolutedev}
\end{figure} 

We repeated this analysis for a number of values of the hopping parameter $\kappa$.
The results are summarized in \cref{fig:absolutedev} for the larger lattice with $\Lambda=64\times 8$. We evaluate the gap between the neural importance sampling (NIS) estimate $\hat F$ and the HMC reference normalized by the total standard error. Namely, if the normalized gap is within the range $[-1,+1]$, both estimators are compatible (see inset in the top plot of \cref{fig:absolutedev}).
Dashed curves connect  q-estimates \eqref{eq:estq}, while solid curves connect p-estimates \eqref{eq:estp}, of the free energy at different values of $\kappa$. The results obtained with $\mathbb{Z}_2$-equivariant reverse-KL and forward-KL flows are shown in blue and orange respectively. The inset shows close agreement of both estimators and both flow models for $\kappa\in\{0.2,0.3\}$. However, deep in the broken phase, e.g. $\kappa\geq0.4$, the two modes of the target distribution start to lay further apart resulting in a failure of the mode-seeking objective, i.e. the reverse-KL, to properly capture the target density, see also bottom left plot of \cref{fig:fwVsrev}. As a result, the probability mass transport induced by the normalizing flow fails to reproduce the correct target distribution $p$, leading to a larger gap between the p- and q-estimators. 
When using the forward-KL trained flow instead, the support of the sampler is closely matching the support of the target hence making free energy compatible with the HMC reference even at the higher values of the hopping parameter $\kappa$, when~\cref{eq:ModeDroppingAssumption} holds. This effect is shown in the inset where there is a good agreement between both estimators of the free energy. This observation suggests that the mode-covering nature of the forward-KL is crucial to ensure that the flow leads to unbiased estimates of physical observables.\par
 
In \cref{fig:absolutedev}, it is also shown that our proposed mode-dropping estimator \eqref{eq:moddropest} correlates well with the observed gap in the free energy estimation. Lastly, we use our estimator $\bar{w}$ to evaluate the support-mismatch of forward and reverse flow models trained at several $\kappa$ values and different lattice sizes as shown in \cref{fig:mdest}. The top and bottom plots refer to lattices of size $\Lambda = 16\times 8$ and $\Lambda = 64\times 8$ respectively. These results demonstrate that the quality of the sampler very quickly deteriorates in the broken phase due to mode-collapse for the model trained by self-sampling. This is not the case for models trained with the forward KL. Indeed, as shown in \cref{fig:mdest}, these models scale significantly better in the volume of the system. 
Furthermore, the non-$\mathbb{Z}_2$ equivariant reverse-KL flow (green), is manifestly mode dropping for $\Lambda=16\times 8$, see \cref{fig:fwVsrev}, with values of $\wbar$ around $0.5$ for values $0.4\leq \kappa\leq 0.7$. This agrees with the left-hand side of \cref{fig:fwVsrev} where only half of the support is covered by the learned variational density in the top row. 
\begin{figure}[t!]
     \centering
     \includegraphics[width=0.48\textwidth]{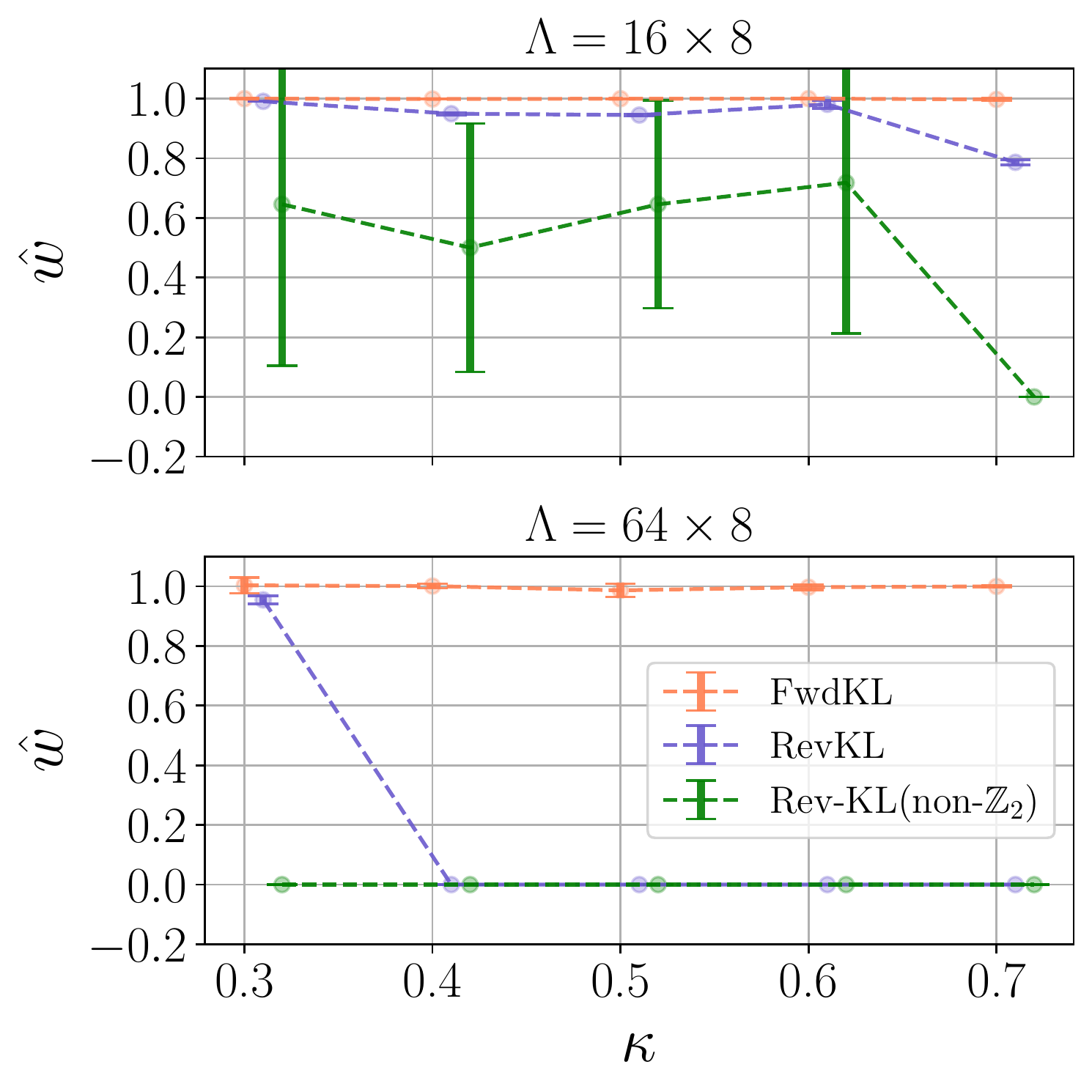}
     \caption{\textbf{Mode dropping evaluation as a function of the hopping parameter $\kappa$ for fixed $\lambda$ and two different lattice sizes} - mode dropping estimator for different values of the hopping parameter and different lattice sizes. For each setup, three normalizing flows were trained with reverse- and forward-KL objectives as described in \cref{sec:fwdVSrev}. The former models are again trained with (orange) and without (green) built-in $\mathbb{Z}_2$ equivariance. The upper and lower plots show the mode-dropping estimator \cref{eq:moddropest} for lattices $\Lambda=16\times 8$ and $\Lambda=64\times 8$, respectively. Unsurprisingly, we observe that mode-collapse gets more severe as the lattice size increases. This is reflected in a stronger decay toward zero of the estimator for larger volumes.}\label{fig:mdest}
\end{figure}
\section{Outlook and Summary}
\label{sec:conclusions}
Mode-collapse presents a significant limitation to flow-based sampling on the lattice because it may lead to inaccurate approximations of the target density, either partially or completely. Intuitively, it can be understood as being in a loose relation to the tunneling problem in local MCMC algorithms. Specifically, the algorithmic challenges in sampling from multi-modal distributions are shifted from the sampling to the training phase for normalizing flows. In this work, we have studied this important limitation of flow-based sampling in great detail. We argue that in the important case of thermodynamic observables, there are practical and theoretically grounded mitigation strategies available. Specifically, the flow can be trained using the forward KL divergence and the free energy can be evaluated with two estimators that bound the true value. Furthermore, we have analyzed mode-mismatch theoretically and derived a bound on its induced bias as well as a quantitative measure for its severity. Normalizing flows are currently only limited to toy models. Encouragingly, we also observed as a side-product of our analysis, that the forward KL objective leads to better scaling in the system size. This observation may be worthwhile to be studied further as part of future work.

\section*{Acknowledgements}
The authors thank the referee for stimulating discussions and useful suggestions that significantly improved the manuscript. K.A.N., C.J.A., S.N., and P.K. are supported by the German Ministry for Education and Research (BMBF) as BIFOLD - Berlin Institute for the Foundations of Learning and Data under the grant BIFOLD23B. K.A.N. has been partially supported by the Einstein Research Unit Quantum (ERU) Project under grant ERU-2020-607. This work is supported with funds from the Ministry of Science, Research, and Culture of the State of Brandenburg within the Centre for Quantum Technologies and Applications (CQTA).
This work is funded by the European Union’s Horizon Europe Framework Program (HORIZON) under the ERA Chair scheme with grant agreement No. 101087126. 
This work is funded by the European Union’s HORIZON MSCA Doctoral Networks programme and the AQTIVATE project (101072344).
The authors acknowledge Lena Funcke and Paolo Stornati for helpful discussions.
\begin{figure}[H]
     \centering
\includegraphics[width=0.35\columnwidth]{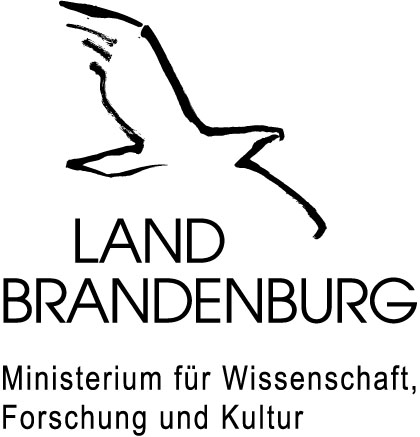}
     \caption*{}
\end{figure}

\bibliography{references}

\clearpage

\appendix
\onecolumngrid
\section*{Appendix}
\subsection{Forward-KL training}
\label{app:forwardtraining}

Training a normalizing flow with forward-KL in the context of lattice field theory requires pre-generated samples at a given point in parameter space. Before training a flow model, one should instantiate a thermalized Markov chain at a fixed value of the coupling parameters and generate a sufficient number of Monte-Carlo configurations which are then used to train the flow. A pseudo-code for this approach is presented in \cref{alg:FKL}. We note that practically this approach may not always be feasible. For example, the number of pre-generated configurations needed for training a flow to an acceptable accuracy increases as the size of the lattice grows. For instance, training a flow for a $64\times 8$ lattice in the context of the $\phi^4$ field theory, in the broken phase, requires already more than fifty million samples. This problem, therefore, limits the practical deployment of forward-KL training schemes at larger scales. Moreover, another limitation of such an approach is that generating samples with HMC may not always be possible. Indeed, in the proximity of a phase transition, long-range autocorrelation will prevent to samples a necessary large amount of uncorrelated samples in time. One would therefore need to be very careful in generating a suitable dataset of HMC configurations to avoid incorporating any additional unwanted bias when training the flow.

\begin{algorithm}[ht]
\SetKwFunction{FSample}{sampleHMC}
\KwIn{
\begin{itemize}[noitemsep,nolistsep]
    \item prior density, e.g., $q_z \sim \normdist (0, \id)$
    \item parametric model with parameters $\theta$
    \item parametric action $S(\phi, \kappa, \lambda)$ with fixed coupling parameters $\lambda$ and $\kappa$
    \item empty tensor for storing a batch of $B$ configurations $\vect{\Phi} \in \realset^{B\times N_S \times N_T}$
\end{itemize}
}
\KwResult{
\begin{itemize}[noitemsep,nolistsep]
    \item learned bijective transformation $f_\theta$ s.t. $\phi_i=f_\theta(z_i)$
    \item exact likelihood function $q_\theta$
\end{itemize}
}
\Begin{
\tcc{Generate training samples from {HMC}}
\For{$c$ in $\{1, \dots, C_{\mathrm{max}}\}$ \label{alg:startsampling}}{
    $\phi$ = \FSample($\lambda$, $\kappa$) \tcp*[l]{sample configurations $\vect{\phi}\in\realset^{n\times N_S \times N_T}$} \label{alg:line1.2}
    $\vect{\Phi}$ = concatenate($\vect{\Phi}$, $\vect{\phi}$) \tcp*[l]{concatenate configurations}
}\label{alg:endsampling}
\tcc{Flow training using the generated dataset $\vect{\Phi}$}
\For{$t$ in $1, \dots, T_{\mathrm{max}}$}{
    \tcc{Iterate over dataset to collect batches of configurations}
    draw samples $\vect{\phi}=\{\phi_i\}_{i=1}^m$ from $\vect{\Phi}$ where $\phi_i\in\realset^{N_S\times N_T}\,\,\forall i\in\{1,\dots,\,m\}$ \label{alg:line3.2}\\
    use $\vect{\phi}$ to evaluate $\frac{1}{m}\sum_{i=1}^m \ln q_\theta(\phi_i)$; \label{alg:line4.2}\\
     $\Delta \theta =  \argminAlg\limits_\theta \estimateE{p}{\ln q_\theta(\phi)} $; \label{alg:line5.2}\\
   update $f_\theta$ with $\theta \leftarrow \theta + \eta \, \Delta \theta$ \label{alg:line6.2};
}
}
\KwRet{$f_\theta$, $q_\theta$}
\caption{To train an NF with forward-{KL} in the context of lattice field theory we need to generate training configurations from a thermalized {HMC} chain. This {HMC} pre-sampling process is made more efficient by running $C_{max}$ independent chains in parallel so that the runtime to sample the entire dataset $\vect{\Phi}$ is constant in the total number of samples $B$. The total number of samples will therefore be $B=n\,C_{max}$. The sampling is done between \cref{alg:startsampling} and \cref{alg:endsampling} in the algorithm below. Once the dataset is sampled and stored on disk, one starts training up to $T_{max}$ iterations. Per iteration, one draws batches of $m$ configurations from $\vect{\Phi}$ in \cref{alg:line3.2} and uses them to evaluate the expectation value of the log-density of $q_\theta$ and compute the gradient of the forward {KL} objective in \cref{alg:line4.2} and \cref{alg:line5.2} respectively. The model weights are then updated and the learned bijection $f_\theta$ along with the variational density $q_\theta$ are returned by the algorithm at the end of its training steps.}\label{alg:FKL}
\end{algorithm}

\subsection{Relative effective support}
\label{app:effectivesupp} 
Let $B(\phi,r)$ be the open ball centered at $\phi$ with radius $r$. A point $\phi$ is called $\epsilon$-dropped if and only if
\begin{align}
    \limsup_{r\to0}\frac{\int_{B(\phi,r)} q_\theta(\phi')d\phi'}{\int_{B(\phi,r)} p(\phi')d\phi'}<\epsilon.
\end{align}
By the Lebesgue differentiation theorem, this implies that $q_\theta(\phi)\le\epsilon p(\phi)$ holds for for almost every $\epsilon$-dropped $\phi$, and if $p$ and $q_\theta$ are continuous, it actually means $q_\theta(\phi)\le\epsilon p(\phi)$. 
We recall the definition of effective relative support \cref{eq:eff_rel_supp} 
\begin{align}
    \widetilde{\mathrm{supp}}_{p,\epsilon}(q_\theta):=\{\phi\in\mathrm{supp}(q_\theta);\ \phi\text{ not $\epsilon$-dropped}\}\,.
\end{align}
Setting $\mathcal{S}:=\textrm{supp}(p)\setminus\widetilde{\mathrm{supp}}_{p,\epsilon}(q_\theta)$ and assuming $\mathcal{S}\ne\emptyset$ means that the importance weighted estimator with $N$ samples lacks a contribution from the mass $\int_{\mathcal{S}}p(\phi)d\phi$ with probability 
\begin{align}
    \left(1-\int_\mathcal{S}q_\theta(\phi)\textrm{d}\phi\right)^N > \left(1-\epsilon\int_{\mathcal{S}}p(\phi)d\phi\right)^N\approx 1 - \epsilon N\int_{\mathcal{S}}p(\phi)d\phi.
\end{align}

\subsection{Proof of \texorpdfstring{\Cref{th:free_en}}{Theorem 1}}
\label{app:theorem1_proof}

\begin{remark}
Suppose the trained model is mode dropping, i.e., 
the approximation~\eqref{eq:ModeDroppingAssumption} holds.
Then the $q$-estimator $\hat{F}_q$ and the $p$-estimator $\hat{F}_p$ for the free energy approximate $\bar{F}_{q}$ and $\bar{F}_{p}$, respectively, being bounds on the true free energy as
\begin{align*}
\bar{F}_{q} \geq F \geq \bar{F}_p \,.
\end{align*}
Furthermore, if $\widetilde{\mathrm{supp}}_{p,\epsilon}(q_\theta) \supseteq \mathrm{supp}(p)$ it follows
\begin{align*}
\bar{F}_{q} = F\,,
\end{align*}
and similarly if $\widetilde{\mathrm{supp}}_{q_\theta,\epsilon}(p) \supseteq \mathrm{supp}(q_\theta)$
\begin{align*}
\bar{F}_{p} = F.\,
\end{align*}
\end{remark}

\begin{proof}
From the definition of the free energy $F=-T\ln Z$ we first note that $\bar{F}_q\ge F$ is equivalent to $\bar{Z}_{q_\theta}\le Z$. Using the fact that $\textrm{supp}(e^{-S(\phi)})=\textrm{supp}(p)$, we obtain~\footnote{To make the notation more compact, in the integrals, we drop the subscript in the effective supports of both $q_\theta$ and $p$.}
\begin{align*}
    \bar{Z}_{q_\theta} &\equiv \estimateE{\widetilde{q}_\theta}{\tilde{w}(\phi)} \\
    &= \int_{\widetilde{\mathrm{supp}}(q_\theta)} \mathcal{D}[\phi]\, q_\theta(\phi) \frac{e^{-S(\phi)}}{q_\theta(\phi)} \\
    &= \int_{\widetilde{\mathrm{supp}}(q_\theta)} \mathcal{D}[\phi]\, e^{-S(\phi)} \\
    &= \int_{\widetilde{\mathrm{supp}}(q_\theta) \cap\, \textrm{supp}(p)} \mathcal{D}[\phi]\, e^{-S(\phi)} \leq \int_{\textrm{supp}(p)}\mathcal{D}[\phi]\,e^{-S(\phi)} = Z\,,
\end{align*}
where the last inequality holds because $e^{-S(\phi)}\geq 0$. Thus, we conclude $\bar{F}_q\ge F$ with the corollary that $\widetilde{\mathrm{supp}}(q_\theta)\supseteq\mathrm{supp}(p)$ implies equality $\bar{F}_q=F$. 

Similarly, 
\begin{align*}
    \bar{Z}_{p}^{-1} &\equiv \estimateE{\widetilde{p}}{\frac{q_\theta(\phi)}{e^{-S(\phi)}}} \\
    &= \frac{1}{Z} \int_{\widetilde{\textrm{supp}}(p)}  \mathcal{D}[\phi]\, e^{-S(\phi)} \frac{q_\theta(\phi)}{e^{-S(\phi)}} \\
    &= \frac{1}{Z} \int_{\widetilde{\textrm{supp}}(p)}  \mathcal{D}[\phi]\, q_\theta(\phi) \\
    &= \frac{1}{Z} \int_{\widetilde{\textrm{supp}}(p) \cap\, {\mathrm{supp}}(q_\theta)}  \mathcal{D}[\phi] \,\underbrace{q_\theta(\phi)}_{\geq 0}
    \leq \frac{1}{Z} \int_{{\mathrm{supp}}(q_\theta)}  \mathcal{D}[\phi]\, q_\theta(\phi) = Z^{-1}
\end{align*}
shows $\bar{F}_p\le F$, in general, and $\bar{F}_p=F$ given $\widetilde{{\mathrm{supp}}}(p)\supseteq{\mathrm{supp}}(q_\theta)$.

Hence, by combining the inequalities we can conclude
\begin{align*}
    \bar{F}_q\ge F\ge \bar{F}_p.
\end{align*}
\end{proof}

\subsection{Bound on the configuration}
\label{app:boundconfig}
Let us assume $S(\phi)$ to be polynomial bounded
\begin{align}\label{eq:bound1}
    C_\vee \vert\vert \phi \vert\vert^{\alpha_\vee} -D_\vee \leq S(\phi) < C_\wedge \vert\vert \phi \vert\vert^{\alpha_\wedge} + D_\wedge
\end{align}
with non negative coefficients $C_\vee,\,C_\wedge,\,D_\vee,\,D_\wedge,\,\alpha_\vee,\,\alpha_\wedge\in \mathbb{R}_{\ge0}$.
The left-hand and right-hand sides represent the lower and the upper bounds on the action $S$.
One needs to find appropriate coefficients $C_{\vee,\,\wedge}, \, \alpha_{\vee,\,\wedge},\, D_{\vee,\,\wedge}$ such that the inequalities are satisfied. 
We now do a foliation of the sampling space 
\begin{equation}\label{eq:foliation}
    \Delta_n := \{\phi \,\, | \,\, e^{(D_\vee +1 -n)} \geq e^{-S(\phi)} > e^{(D_\vee - n)}\}
\end{equation}
which can be seen as a re-distribution of the lattice configurations $\phi$ into infinitely many buckets labeled by the index $n$.
Combining \cref{eq:bound1} and \cref{eq:foliation} one can rewrite a condition on the norm of $\phi$~\footnote{We take the $2$-norm for the field configuration and drop the subscript for notation convenience. We implicitly assume $\vert\vert\phi\vert\vert\equiv\vert\vert\phi\vert\vert_2$.}. For a configuration $\phi \in \Delta_n$
\begin{align}
    C_\vee \vert\vert \phi \vert\vert^{\alpha_\vee} -D_\vee \leq S(\phi) < - D_\vee + n\,.
\end{align}
This implies
\begin{align}\label{eq:bound11}
    C_\vee \vert\vert \phi \vert\vert^{\alpha_\vee} <  n \implies \vert\vert \phi \vert\vert < \left(\frac{n}{C_\vee}\right)^\frac{1}{\alpha_\vee}
\end{align}
for $n>0$. On the other side, it follows that
\begin{align}
    C_\wedge \vert\vert \phi \vert\vert^{\alpha_\wedge} + D_\wedge \geq S(\phi) \geq -D_\vee -1 + n\,,
\end{align}
which implies that 
\begin{align}\label{eq:bound22}
    C_\wedge \vert\vert \phi \vert\vert^{\alpha_\wedge} + D_\wedge \geq -D_\vee -1 + n \implies \vert\vert \phi \vert\vert \geq \left(\frac{-D_\vee -D_\wedge -1 +n}{C_\wedge}\right)^\frac{1}{\alpha_\wedge}\,
\end{align}
for $C_\wedge>0,\, n>0$.  Combining \cref{eq:bound11} and \cref{eq:bound22} we obtain the following bounds on the norm of the lattice configuration
\begin{align}\label{eq:bigbound}
    \underbrace{\left[\frac{-D_\vee - D_\wedge -1 + n}{C_\wedge}\right]^{\frac{1}{\alpha_\wedge}}}_{\textrm{l}_n} \leq \vert\vert\phi\vert\vert< \underbrace{\left(\frac{n}{C_\vee}\right)^{\frac{1}{\alpha_\vee}}}_{\textrm{u}_n}\,.
\end{align}
In particular, this can be shown to imply a bound on the volume of $\Delta_n$, e.g., making the volume finite
\begin{align}\label{eq:bound_vol}
    \textrm{vol}(\Delta_n) \le \int_{l_n\leq\vert\vert\phi\vert\vert<u_n}d\phi = \frac{\pi^\frac{N}{2}}{\Gamma\left(\frac{N}{2} + 1\right)}\cdot \left(u_n^N-l_n^N\right)\,,
\end{align}
where we used the volume of the $N-$ball, 
\begin{align*}
    \frac{\pi^\frac{N}{2}}{\Gamma\left(\frac{N}{2} + 1\right)}\cdot r^N\,.
\end{align*}

\subsection{Proof of \texorpdfstring{\Cref{th:bias}}{Theorem 1}}
\label{app:theorem2_proof}
We now leverage the result from \cref{app:boundconfig} to derive a bound on the bias for the general observable $\mathcal{O}$ when $\widetilde{\textrm{supp}}({q}_\theta) \subset \textrm{supp}(p)$ and therefore the importance sampling estimator may not be unbiased when~\cref{eq:ModeDroppingAssumption} holds true.
\begin{remark}
Let the action $S$ of the theory and the observable $\mathcal{O}$ be polynomially bounded, i.e.
\begin{align}
    C_\vee \vert\vert \phi \vert\vert^{\alpha_\vee} -D_\vee \leq S(\phi) < C_\wedge \vert\vert \phi \vert\vert^{\alpha_\wedge} + D_\wedge
\end{align}
for $C_{\vee, \wedge}, D_{\vee, \wedge}, \alpha_{\vee, \wedge} \in \mathbb{R}_{\ge0}$ and
\begin{align}
    \vert\mathcal{O}(\phi)\vert \leq c\vert\vert\phi\vert\vert^\alpha + d
\end{align}
for $c,\alpha,d \in \mathbb{R}_{\ge0}$.
Then, the bias between the estimated observable $\hat{\mathcal{O}}$ and the true value $\mathcal{O}^*$ is given by
\begin{align}\label{eq:bias_app}
    \vert \bar{\mathcal{O}} - \mathcal{O}^*\vert \leq \sum_{n\in N}  \sup_{\phi\in\Delta_n}\vert\mathcal{O}(\phi)\vert \cdot \alpha_n\,,
\end{align}
where $\bar{\mathcal{O}} = \mathbb{E}_{\phi \sim q_\theta} \left[ \hat{\mathcal{O}}(\phi) \right]$.
\end{remark}
\begin{proof}
Let's assume the generic observable $\mathcal{O}$ to be polynomially bounded
\begin{align}\label{eq:gen_bound}
    \vert\mathcal{O}(\phi)\vert \leq c\vert\vert\phi\vert\vert^\alpha + d\,.
\end{align}

When $\widetilde{\textrm{supp}}({q}_\theta) \subseteq \textrm{supp}(p) $, \cref{eq:bias_app} can be bounded as
\begin{align}
    \vert \bar{\mathcal{O}} - \mathcal{O}^*\vert &= \left\vert \int_{\widetilde{\textrm{supp}}({q}_\theta)} \Dphi \mathcal{O}(\phi)\frac{p(\phi)}{q_\theta(\phi)}q_\theta(\phi)  - \int \Dphi \mathcal{O}(\phi) \,p(\phi)\right\vert\\
    &= \left\vert \int \left( 1 - 1_{\widetilde{\textrm{supp}}({q}_\theta)} (\phi) \right) \,\mathcal{O}(\phi)\,p(\phi)\,\Dphi\right\vert\,\\
    &\leq \sum_{n\in \mathbb{N}} \left\vert \int_{\Delta_n} \mathcal{O}(\phi) \left( 1 - 1_{\widetilde{\textrm{supp}}({q}_\theta)} (\phi) \right) p(\phi) \, \Dphi \right\vert \\
    \leq& \sum_{n\in \mathbb{N}} \sup_{\phi\in\Delta_n}\left\vert\mathcal{O}(\phi)\right\vert \cdot \underbrace{\int_{\Delta_n} \left( 1 - 1_{\widetilde{\textrm{supp}}({q}_\theta)} (\phi) \right) p(\phi) \, \Dphi}_{=:\alpha_n}\,,
\end{align} 
where $\alpha_n$ is the coefficient defined for each bucket, i.e.,
\begin{align}
\alpha_n =& \int_{\Delta_n} \left( 1 - 1_{\widetilde{\textrm{supp}}({q}_\theta)} (\phi) \right) p(\phi) \, \Dphi \nonumber \\
=& \frac{1}{Z}\int_{\Delta_n} \left( 1 - 1_{\widetilde{\textrm{supp}}({q}_\theta)} (\phi) \right) e^{-S(\phi)}  \, \Dphi \nonumber \\ 
=& \frac{1}{Z}\left(\int_{\Delta_n} e^{-S(\phi)}  \, \Dphi -  \int_{{\Delta_n}\cap\,{\widetilde{\textrm{supp}}({q}_\theta)}} e^{-S(\phi)}  \, \Dphi \right)\,, \label{eq:w_n}
\end{align}
such that the following relation holds
\begin{align}\label{eq:5}
    \sum_{n\in \mathbb{N}} \alpha_n = \int \left( 1 - 1_{\widetilde{\textrm{supp}}({q}_\theta)} (\phi) \right) p(\phi) \, \Dphi =   1 - \wbar\,.
\end{align}
One, therefore, concludes that the bias is bounded by the following series
\begin{align}\label{eq:bias}
    \vert \bar{\mathcal{O}} - \mathcal{O}^*\vert \leq& \sum_{n\in \mathbb{N}}  \sup_{\phi\in\Delta_n}\left\vert\mathcal{O}(\phi)\right\vert \cdot \alpha_n\tobias{.}
\end{align}

In order to obtain convergence of this series, we observe that polynomial boundedness of the observable implies 
\begin{align*}
     \sup_{\phi\in\Delta_n}\left\vert\mathcal{O}(\phi)\right\vert\le c\,u_n^\alpha+d\,,
\end{align*}
i.e., $\sup_{\phi\in\Delta_n}\left\vert\mathcal{O}(\phi)\right\vert$ grows polynomially in $n$. Similarly, from \cref{eq:w_n}, it follows
\begin{align}
    \alpha_n\le \frac{1}{Z} \textrm{vol}(\Delta_n)e^{D_\vee+1-n}\,,
\end{align} 
showing that $\alpha_n$ decays exponentially in $n$. Thus, $ \sup_{\phi\in\Delta_n}\left\vert\mathcal{O}(\phi)\right\vert \cdot \alpha_n$ decays exponentially in $n$. This implies convergence of the series in the right-hand side of \cref{eq:bias}.

It is important to note that each $\alpha_n$ is weighted by the maximum of the observable on the corresponding volume $\Delta_n$ which makes the bias inherently dependent on the observable while the $\alpha_n$ coefficients are universal and represent the amount of mode dropping per bucket.
\end{proof}

As an explicit example, let us consider $S:\ \mathbb{R}\to\mathbb{R}; S(\phi)=\phi^2$ and the observable $\phi^3$. This means that the true value of the observable is
\begin{align}
\mathcal{O}^*=\mathbb{E}_{\phi \sim p}\left[\mathcal{O}(\phi) \right]=\frac{1}{Z}\int_{-\infty}^{\infty} D[\phi] \phi^3\,e^{-\phi^2} =0\,.
\end{align}
If we assume a mode dropping model $q$, with $q=2p\cdot1_{\mathbb{R}_{\geq 0}}$, then
\begin{align*}
    \vert \bar{\mathcal{O}} - \mathcal{O}^*\vert=\bar{\mathcal{O}}=\frac{2}{\sqrt\pi}\int_0^\infty D[\phi]\phi^3e^{-\phi^2}=\frac12.
\end{align*}
For the definition of the $\Delta_n$, we can choose $D_\vee=1$ and thus $\Delta_n=[-\sqrt{n},-\sqrt{n-1}]\cup[\sqrt{n-1},\sqrt{n}]$ and obtain $\alpha_n=0.5(\mathrm{erf}(\sqrt{n})-\mathrm{erf}(\sqrt{n-1}))$. We note that since $u_n=\sup\{|\phi|,\ \phi\in\Delta_n\}=\sqrt n$ and the observable is $\phi^3$ it follows that $$
\sup\{\vert\mathcal{O}(\phi)\vert,\ \phi\in\Delta_n\}\le(\sqrt n)^3\,.
$$ 
Hence, the bias is within the bound given by the theorem
\begin{align*}
    \frac12=\vert \bar{\mathcal{O}} - \mathcal{O}^*\vert\le\sum_{n\in\mathbb{N}}n^{3/2}\alpha_n\approx 0.73.
\end{align*}

\subsection{Details on the Numerical Experiments}
\label{app:technical_details}
In the following, we summarize the details and setup used to perform the training of both forward- and reverse-KL normalizing flows as well as to estimate the HMC reference values. For our experiments, we focused on the action $S(\phi)$ from \cref{eq:action} as a function of $\kappa$ while keeping the coupling $\lambda=0.022$ fixed throughout the analysis. 
\subsubsection{HMC sampling}\label{app:hmc_est}
For estimating the HMC reference values of the free energy density reported in \cref{sec:results}, we followed the same approach as in \cite{nicoli2021estimation}. The idea is to discretize the trajectory in the $\kappa$-space (hopping parameter) into a sequence of finite steps where free energy differences can be calculated by running an HMC. The target free energy at an arbitrary point $\kappa^*$ is then obtained by summing up all the free energy differences from $\kappa=0$ to $\kappa=\kappa^*$. We note that the higher the kappa values, the more steps one needs to make in order to discretize the trajectory up to the target point in parameter space. It follows that the uncertainty on the estimates also grows when $\kappa^*$ increases as more terms are combined to obtain the free energy at the desired $\kappa^*$. 
Specifically, in our experiments, we chose a regular step-size between two subsequent $\kappa$ values in the trajectory to be $\Delta\kappa=0.01$. Such step-size is used to discretize the trajectory starting from $\kappa=0.0$ all the way up to the target. For instance, measuring the free energy density at $\kappa^*=0.3$ would therefore require thirty steps, hence 30 independent HMC chains. Each of these chains is initialized around the vacuum expectation value (vev), has an overrelaxation every 10 steps, and a total of 10k thermalization steps, i.e. discarded configuration updates, followed by 500k sampling steps. Those configurations, from the equilibrium distribution, are used to estimate the free energy difference at a single given point of the trajectory.
The total number of HMC samples needed to estimate the free energy at an arbitrary point thus depends on $\kappa^*$. For instance, referring to the previous example, 30 chains with 500k steps each add up to 15M HMC samples.

\subsubsection{Reverse-KL Flow}
To train the reverse-KL normalizing flows we followed the same strategy presented in \cite{nicoli2021estimation} with the same setup of hyperparameters. We used a batch size of 8k samples and a learning rate update according to the ReduceLROnPlateau scheduler of PyTorch with an initial learning rate of $5 \times 10^{-4}$ and patience of 3k steps. The flows have the same number of coupling blocks and the same type of checkerboard partitioning discussed in \cite{nicoli2021estimation}. Models were trained for 700k steps in total and the last saved checkpoint is used for sampling. Every reverse-KL model was trained on two GPUs (in parallel), either P100 or A100 NVIDIA devices. Depending on the lattice volume and the model type the training took up to ~50hrs of wall time.

\subsubsection{Forward-KL Flow}
Training a forward-KL flow requires a different procedure which was outlined in \cref{app:forwardtraining}. For every flow model, we used 50M pre-sampled HMC configurations as input data. These were sampled in batches of 100 independent HMC chains each of which had 10k equilibration (discarded) and 500k sampling (stored) steps. The stored configurations from each chain in the batch were concatenated to generate the full training set. \par
At the stage of training, the 50M configurations are loaded in batches of 8k samples per iteration (training step). When the entire dataset is processed once, the full set of configurations is reshuffled and reused (as it is standard practice in deep learning) until the desired number of training iterations is reached. Again, forward-KL models were trained for 700k steps on two GPUs (in parallel), either P100 or A100 NVIDIA devices. Depending on the lattice volume and the model type the training took up to ~55hrs of wall time.

\subsubsection{Flow sampling}
For sampling configurations from both forward- and reverse-KL normalizing flows we proceed as follows. In order to have a fair comparison with HMC one would need to sample as many configurations as those needed to integrate the trajectory in the hopping parameter space, as discussed in \cref{app:hmc_est}. However, for our flow estimates, we took only 1M configurations and used the estimators for mean and variance introduced in \cite{nicoli2020asymptotically,nicoli2021estimation} and proposed in \cref{sec:F_estimators}. Though 1M is in general a lower bound on the total amount of configurations used to compute HMC estimates, we empirically observed this was sufficient to obtain estimates with errors several orders of magnitude smaller than HMC. Therefore, we took this as a sufficient number of samples for comparing the two sampling approaches.
\end{document}